\documentclass[12pt]{article}
\usepackage{a4wide}
\usepackage{latexsym}
\usepackage{amsmath}
\usepackage{amsfonts}
\usepackage{amscd}
\usepackage{amsthm}
\usepackage{cite}
\usepackage{shuffle}

\usepackage{pslatex}
\usepackage{graphicx}
\usepackage[latin1,utf8]{inputenc}
\usepackage[T1]{fontenc}

\newcommand{\bq}{\begin{eqnarray}}
\newcommand{\eq}{\end{eqnarray}}
\newcommand{\eps}{\varepsilon}
\newcommand{\qbar}{\bar{q}}
\newcommand{\slashoperator}[2]{|_{#2} #1}

\theoremstyle{plain}
\newtheorem{theoremcounter}{}[]
\newtheorem{lemma}[theoremcounter]{Lemma}

\allowdisplaybreaks

\begin{document}

\thispagestyle{empty}

\begin{flushright}
  MITP/22-070
% \\ version of \today
\end{flushright}

\vspace{1.5cm}

\begin{center}
  {\Large\bf The three-loop equal-mass banana integral in $\varepsilon$-factorised form with meromorphic modular forms \\
  }
  \vspace{1cm}
  {\large Sebastian P\"ogel, Xing Wang and Stefan Weinzierl \\
  \vspace{1cm}
      {\small \em PRISMA Cluster of Excellence, Institut f{\"u}r Physik, }\\
      {\small \em Johannes Gutenberg-Universit{\"a}t Mainz,}\\
      {\small \em D - 55099 Mainz, Germany}\\
  } 
\end{center}

\vspace{2cm}

% abstract ---------------------------------------
\begin{abstract}\noindent
  {
We show that the differential equation for the three-loop equal-mass banana integral can be cast 
into an $\varepsilon$-factorised form with entries constructed from 
(meromorphic) modular forms and one special function, 
which can be given as an iterated integral of meromorphic modular forms. 
The $\varepsilon$-factorised form of the differential equation allows for a systematic solution 
to any order in the dimensional regularisation parameter $\varepsilon$.
The alphabet of the iterated integrals contains six letters. 
   }
\end{abstract}

\vspace*{\fill}

% -----------------------------------------------------------------------------
\newpage

\section{Introduction}
\label{sect:intro}

For precision calculations in high-energy particle physics we need to evaluate Feynman integrals.
It is a very interesting question which special functions appear in the final result of such a calculation.
In the simplest case we find multiple polylogarithms \cite{Goncharov_no_note,Goncharov:2001,Borwein,Moch:2001zr}, which are iterated integrals of differential one-forms
$dy/(y-c)$, where $c$ is a constant.
An example is the one-loop bubble integral with equal non-zero masses.
The equal-mass one-loop bubble integral can be expressed to all orders in the dimensional regularisation parameter $\eps$
in terms of multiple polylogarithms.
A second non-trivial example is given by the massless two-loop double box integrals \cite{Smirnov:1999gc,Tausk:1999vh}.

More complicated Feynman integrals are related to elliptic curves and evaluate to iterated integrals of modular forms
and/or specific differential one-forms related to the Kronecker function.
The most prominent example is the two-loop sunrise integral with equal non-zero masses.
The equal-mass two-loop sunrise integral can be expressed to all orders in the dimensional regularisation parameter $\eps$
in terms of iterated integrals of modular forms.

In this paper we consider the three-loop banana integral with equal non-zero masses.
The three-loop banana integral is the next more complicated integral in the family of $l$-loop banana integrals.
The first two members of this family are 
the one-loop bubble integral
and the two-loop sunrise integral.
The geometry of the three-loop banana integral is related to a Calabi-Yau $2$-fold (i.e. a surface) and thus more complicated than 
an elliptic Feynman integral, which in this context can be viewed as related to a Calabi-Yau $1$-fold (i.e. a curve).

A convenient tool to obtain the result for a Feynman integral to any desired order 
in the dimensional regularisation parameter $\eps$ is the method of differential equations \cite{Kotikov:1990kg,Kotikov:1991pm,Remiddi:1997ny,Gehrmann:1999as}.
For example, this method has been used at the beginning of the millennium 
to obtain the two-loop master integrals for $\gamma^\ast \rightarrow 3 \; \mbox{jets}$ \cite{Gehrmann:2000zt,Gehrmann:2001ck}.
The solution of the differential equation is significantly simplified, if the differential equation is in $\eps$-factorised form \cite{Henn:2013pwa}.
If the differential equation is in $\eps$-factorised form, we may read off the letters appearing in the iterated integrals from the differential
equation.
Such a form for the differential equation has been found for many Feynman integrals evaluating to multiple polylogarithms
and selected examples of elliptic Feynman integrals \cite{Adams:2018yfj,Bogner:2019lfa,Muller:2022gec}.
In this paper we present the differential equation for the equal-mass three-loop banana integral in $\eps$-factorised form.

The three-loop banana integral has been studied in the literature in the past \cite{Bloch:2014qca,Primo:2017ipr,Broedel:2019kmn,Broedel:2021zij,Klemm:2019dbm,Bonisch:2020qmm,Kreimer:2022fxm} and we should carefully explain what is new in this article.

Let $I$ be a vector of master integrals. The differential equation 
\bq
 d I & = & \eps A I
\eq
is said to be in $\eps$-factorised form, if the connection matrix $A$ is independent of $\eps$.
An $\eps$-factorised form can be achieved, if a full set of homogeneous solutions is known.
The set of homogeneous solutions can be obtained by integrating the integrands of the master integrals over
a set of independent contours \cite{Primo:2017ipr}.
This path has been followed in refs.~\cite{Primo:2017ipr,Frellesvig:2021hkr} and an $\eps$-factorised form of the
differential equation follows from the results of these papers.
However, this is not the form we are interested in.
It can be shown in the two-loop sunrise case that the solution for the Feynman integrals obtained in this way is not
of uniform weight.

The first term of the $\eps$-expansion of the three-loop banana integral 
(normalised to a homogeneous solution of the Picard-Fuchs operator)
is known from ref.~\cite{Bloch:2014qca}.
It corresponds to an Eichler integral
\bq
\label{Eichler_integral_banana}
 \left[ I\left(1,1,f_4;\tau\right) + \frac{4}{3} \zeta_3 \right] \eps^3,
\eq
where $f_4$ is a modular form for $\Gamma_1(6)$ of modular weight $4$. The notation for iterated integrals is defined in sect.~\ref{sect:iterated_integrals}.
In the main part of the paper $f_4$ will be denoted as $f_{4,a}$.
This mirrors closely the first term of the 
$\eps$-expansion of the two-loop sunrise integral 
(again normalised to a homogeneous solution of the appropriate Picard-Fuchs operator) \cite{Adams:2017ejb,Adams:2018yfj}
\bq
 \left[ 4 I\left(1,f_3;\tau\right) + 3 \mathrm{Cl}_2\left(\frac{2\pi}{3}\right)\right] \eps^2,
\eq
where $f_3$ is a modular form for $\Gamma_1(6)$ of modular weight $3$
and $I(1,f_3;\tau)$ is again an Eichler integral. $\mathrm{Cl}_2$ denotes the Clausen function.
We are interested in the higher-order terms in the dimensional regularisation parameter $\eps$ which add to eq.~(\ref{Eichler_integral_banana}).
This has been considered in ref.~\cite{Broedel:2021zij}, where it was shown that the higher-order terms are given
by iterated integrals of meromorphic modular forms.
The authors of this reference also gave a differential equation. 
However this differential equation is not $\eps$-factorised,
as it contains an additional $1/\eps$-term.
In this paper we improve the situation by deriving a differential equation in $\eps$-factorised form.

At first sight it seems surprising that the three-loop banana integral can be expressed in terms
of iterated integrals of (meromorphic) modular forms as the geometry is related to a Calabi-Yau $2$-fold and not
an elliptic curve.
However it has been known for a long time that the corresponding Picard-Fuchs operator (a third-order differential operator)
is the symmetric square of a second-order differential operator \cite{Verrill:1996,Joyce:1972}.
By a variable transformation, this second-order differential operator is related to the Picard-Fuchs operator
of the two-loop sunrise integral \cite{Bloch:2014qca,Broedel:2019kmn}.
We start with the dimensionless kinematic variable $x=p^2/m^2$.
The variable transformation maps the pseudo-threshold $x=4$ and the threshold $x=16$ of the three-loop banana integral
to the points $y=3$ and $y=-3$, respectively.
The variable $y$ is the natural ``physical'' variable in the sunrise context.
Mapped to the complex upper half-plane, these points correspond to $\tau=\frac{1}{4}(1+i\sqrt{3})$ and
$\tau = \frac{1}{6}(3+i\sqrt{3})$, respectively.
At these points the elliptic curve of the two-loop sunrise integral is not degenerate.
As the original differential equation for the three-loop banana integral has singularities at these points,
meromorphic modular forms (with poles at $\tau=\frac{1}{4}(1+i\sqrt{3})$ and $\tau = \frac{1}{6}(3+i\sqrt{3})$)
emerge naturally \cite{Matthes:1972,Broedel:2021zij}.
We show that the $\eps$-factorised differential equation
has at most simple poles at these points.

This paper is organised as follows:
In section~\ref{sect:defintions} we define the three-loop banana integral and the associated Picard-Fuchs operator.
In section~\ref{sect:master_integrals} we construct from an ansatz a set of master integrals which lead
to an $\eps$-factorised differential equation.
This ansatz follows closely the steps taken in ref.~\cite{Adams:2018yfj,Bogner:2019lfa,Muller:2022gec}.
In section~\ref{sect:periods} we define the periods for the elliptic curve of the two-loop sunrise integral.
From these periods we can construct a solution for the Picard-Fuchs operator of the three-loop banana integral.
In section~\ref{sect:modular} we introduce (meromorphic) modular forms.
We only need four (meromorphic) modular forms of modular weight $1$ as basic building blocks, which we label
\bq
 b_0,
 \;\;\;
 b_1,
 \;\;\;
 b_3,
 \;\;\;
 b_{-3}.
\eq
All other occurring (meromorphic) modular forms are polynomials in these.
In section~\ref{sect:iterated_integrals} we introduce the notation for iterated integrals
of (meromorphic) modular forms.
In section~\ref{sect:dgl} we present the differential equation for the master integrals in $\eps$-factorised form.
The differential equation contains six letters, which we denote as
\bq
 1,
 \;\;\;
 f_{2,a},
 \;\;\;
 f_{2,b},
 \;\;\;
 f_{4,a},
 \;\;\;
 f_{4,b},
 \;\;\;
 f_{6}.
\eq
The transformation laws of the entries of the differential equation under modular transformations of $\Gamma_1(6)$ are discussed in section~\ref{sect:modular_transformation}.
The differential equation is solved in section~\ref{sect:results} and analytical results for all master integrals up to order $\eps^4$
are presented.
In section~\ref{sect:numerics} we give numerical results and show that they agree with results obtained
from SecDec \cite{Carter:2010hi,Borowka:2017idc,Borowka:2018goh}.
Our conclusions are presented in section~\ref{sect:conclusions}.
In appendix~\ref{sect:boundary} we derive the boundary conditions necessary for solving the differential equation.
Appendix~\ref{sect:supplement} describes the supplementary electronic file attached to this article, which gives the 
solution for the banana integrals up to order $\eps^6$.

% -----------------------------------------------------------------------------

\section{Definitions}
\label{sect:defintions}

\begin{figure}
\begin{center}
\includegraphics[scale=1.0]{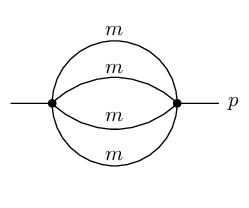}
\end{center}
\caption{
The three-loop banana graph.
}
\label{fig_banana_graph}
\end{figure}
We are interested in the integrals
\bq
\label{def_banana_loop_integral}
 I_{\nu_1 \nu_2 \nu_3 \nu_4}
 & = &
 e^{3 \gamma_E \eps}
 \left(m^2\right)^{\nu-\frac{3}{2}D}
 \int \left( \prod\limits_{a=1}^{4} \frac{d^Dk_a}{i \pi^{\frac{D}{2}}} \right)
 i \pi^{\frac{D}{2}} \delta^D\left(p-\sum\limits_{b=1}^{4} k_b \right)
 \left( \prod\limits_{c=1}^{4} \frac{1}{\left(-k_c^2+m^2\right)^{\nu_c}} \right)
\eq
where $D$ denotes the number of space-time dimensions, $\eps$ the dimensional regularisation parameter,
$\gamma_E$ the Euler-Mascheroni constant
and the quantity $\nu$ is defined by
\bq
 \nu & = &
 \sum\limits_{j=1}^4 \nu_j.
\eq
We consider these integrals in $D=2-2\eps$ space-time dimensions.
It is convenient to introduce the dimensionless variable
\bq
 x & = & \frac{p^2}{m^2}.
\eq
It is well-known that this family of Feynman integrals has four master integrals.
A possible choice for a basis of master integrals is
\bq
 I_{1110}, 
 \;
 I_{1111}, 
 \;
 I_{1112}, 
 \;
 I_{1113}.
\eq
The Feynman integral $I_{1110}$ is a product of three one-loop tadpole integrals
and rather simple.
We set
\bq
 I_1
 & = & 
 \eps^3 I_{1110}
 \; = \; 
 \left[ e^{\gamma_E \eps} \Gamma\left(1+\eps\right) \right]^3. 
\eq
The integral $I_1$ has uniform weight.
For $I_{1111}$ we have the inhomogeneous third-order differential equation
\bq
 L_3 \; I_{1111}
 & = &
 - \frac{24 \eps^3}{x^2\left(x-4\right)\left(x-16\right)} I_{1110},
\eq
with
\bq
 L_3 
 & = &
  \frac{d^3}{dx^3}
  + \left[ \frac{3}{x}
           + \frac{3\left(1+2\eps\right)}{2\left(x-4\right)}
           + \frac{3\left(1+2\eps\right)}{2\left(x-16\right)}
    \right] \frac{d^2}{dx^2}
 \nonumber \\
 & &
  + \left[ \frac{7x^2-68x+64}{x^2\left(x-4\right)\left(x-16\right)}
           + \frac{6 \eps \left(3x-20\right)}{x\left(x-4\right)\left(x-16\right)}
           + \frac{\eps^2 \left(11x+16\right)}{x^2\left(x-16\right)}
    \right] \frac{d}{dx}
 \nonumber \\
 & &
  + \left(1+2\eps\right)\left(1+3\eps\right)
    \left[ \frac{1}{x^2\left(x-16\right)} + \frac{\eps\left(x+2\right)}{x^2\left(x-4\right)\left(x-16\right)} \right].
\eq
We denote by $L_3^{(0)}$ the $\eps$-independent part of $L_3$ \cite{MullerStach:2012mp}:
\bq
 L_3^{(0)}
 = 
  \frac{d^3}{dx^3}
  + \left[ \frac{3}{x}
           + \frac{3}{2\left(x-4\right)}
           + \frac{3}{2\left(x-16\right)}
    \right] \frac{d^2}{dx^2}
  + \frac{7x^2-68x+64}{x^2\left(x-4\right)\left(x-16\right)} \frac{d}{dx}
  + \frac{1}{x^2\left(x-16\right)}.
\eq
Let
\bq
\label{eq_psi_L2}
 L_2^{(0)}
 & = &
  \frac{d^2}{dx^2}
 + \left[\frac{1}{x} + \frac{1}{2\left(x-4\right)} + \frac{1}{2\left(x-16\right)} \right] \frac{d}{dx}
 + \frac{\left(x-8\right)}{4x\left(x-4\right)\left(x-16\right)}.
\eq
The differential operator $L_3^{(0)}$ is the symmetric product of $L_2^{(0)}$ \cite{Verrill:1996}: 
If $\psi_1, \psi_2$ are two independent solutions of $L_2^{(0)}$
\bq
\label{eq_psi_L2_psi}
 L_2^{(0)} \; \psi_i & = & 0,
 \;\;\;\;\;\; i \; \in \; \{1,2\},
\eq
then the solution space of $L_3^{(0)}$ is spanned by
\bq
 \psi_1^2, \;\; \psi_1 \psi_2, \;\; \psi_2^2.
\eq
The Wronskian is defined by
\bq
 W
 & = &
 \psi_1 \frac{d}{dx} \psi_2
 -
 \psi_2 \frac{d}{dx} \psi_1.
\eq
We have
\bq
 \left[ \frac{d}{dx} + \frac{1}{x} + \frac{1}{2\left(x-4\right)} + \frac{1}{2\left(x-16\right)} \right] W & = & 0
\eq
and
\bq
\label{eq_Wronskian}
 W & = & \frac{2 \pi i c}{x \left(4-x\right)^{\frac{1}{2}} \left(16-x\right)^{\frac{1}{2}}},
\eq
where $c$ is a constant.
The constant $c$ depends on the normalisation of the two independent solutions $\psi_1$ and $\psi_2$.
In section~\ref{sect:periods} we give explicit expressions for the two independent solutions of eq.~(\ref{eq_psi_L2_psi}).
This choice of solutions yields $c=3$.

% ------------------------------------------------------------------------------------------
\section{Master integrals}
\label{sect:master_integrals}

In this section we determine a choice of master integrals, which put the differential equation
into an $\eps$-form.
We do this in two steps:
In the first step we start from a rather general ansatz with five unknown functions.
We derive differential equations these functions have to satisfy from the condition that the
differential equation for the master integrals is in $\eps$-form.
Four of these five functions are easily solved for, leading to a more specific ansatz, which 
we consider in step $2$. This specific ansatz involves only one unknown function.
We determine this remaining function in section~\ref{sect:master_integrals_step_2}.

\subsection{Step 1}
\label{sect:master_integrals_step_1}

Let $\omega_1$, $J$, $F_{32}$, $F_{42}$ and $F_{43}$ be (a priori unknown) functions of $x$.
Here, $J$ denotes the Jacobian for a (a priori unknown) change of variable from $x$ to $\tau$:
\bq
\label{def_Jacobian}
 J & = & \frac{dx}{d\tau}.
\eq
This implies
\bq
 \frac{d}{d\tau} & = & J \frac{d}{dx}.
\eq
We start from the following ansatz for the 
master integrals
\bq
 I_1
 & = &
 \eps^3 I_{1110},
 \nonumber \\
 I_2
 & = &
 \eps^3 \frac{\pi^2}{\omega_1} I_{1111},
 \nonumber \\
 I_3
 & = &
 \frac{1}{2\pi i \eps} \frac{d}{d\tau} I_2
 + F_{32} I_2,
 \nonumber \\
 I_4
 & = &
 \frac{1}{2\pi i \eps} \frac{d}{d\tau} I_3
 + F_{42} I_2 
 + F_{43} I_3.
\eq
The differential equation for this set of master integrals is in $\eps$-form
\bq
 d I & = & \eps A I,
\eq
provided the following set of six differential equations for the functions
$\omega_1$, $J$, $F_{32}$, $F_{42}$, $F_{43}$ hold:
For $\omega_1$ we require the two equations
\bq
\label{eq_omega_L3}
 & & L_3^{(0)} \; \omega_1 = 0,
 \\
\label{eq_omega_v1}
 & &
 \frac{1}{\omega_1} \frac{d^2\omega_1}{dx^2}
 - \frac{1}{2} \left( \frac{1}{\omega_1} \frac{d\omega_1}{dx} \right)^2
 + \frac{2\left(x^2-15x+32\right)}{x\left(x-4\right)\left(x-16\right)} \frac{1}{\omega_1 } \frac{d\omega_1}{dx}
 + \frac{\left(x-8\right)}{2x\left(x-4\right)\left(x-16\right)}
 = 
 0.
\eq
For $J$ we require
\bq
\label{eq_dgl_Jacobian}
 \frac{d \ln J}{dx}
 & = &
 \frac{d \ln \omega_1}{dx}
 +\frac{2\left(x^2-15x+32\right)}{x\left(x-4\right)\left(x-16\right)}.
\eq
For $F_{32}$, $F_{42}$, $F_{43}$ we require
\bq
\lefteqn{
 \frac{d^2F_{32}}{dx^2}
 + \left[ \frac{d \ln \omega_1}{dx} + \frac{2\left(x^2-15x+32\right)}{x\left(x-4\right)\left(x-16\right)} \right] \frac{dF_{32}}{dx}
 + \frac{3 J}{2\pi i} 
   \left[
           - \frac{\left(x-10\right)}{\left(x-4\right)\left(x-16\right)} \left( \frac{d \ln \omega_1}{dx} \right)^2
 \right. } & & \nonumber \\
 & & \left.
           - \frac{2\left(x^3-30x^2+228x-640\right)}{x\left(x-4\right)^2\left(x-16\right)^2} \frac{d \ln \omega_1}{dx}
           - \frac{\left(x^3-28x^2+168x-384\right)}{x^2\left(x-4\right)^2\left(x-16\right)^2}
   \right]
 = 
 0,
 \nonumber \\
\lefteqn{
 \frac{dF_{42}}{dx}
 - 3 F_{32} \frac{dF_{32}}{dx}
 + \frac{3J}{2\pi i} \frac{2\left(x-10\right)}{\left(x-4\right)\left(x-16\right)} \frac{dF_{32}}{dx}
 } & & \nonumber \\
 & &
 + \frac{3J}{2\pi i} \left[ 
                            \frac{2\left(x-10\right)}{\left(x-4\right)\left(x-16\right)} \frac{d \ln \omega_1}{dx}
                            + \frac{2\left(x^3-30x^2+228x-640\right)}{x\left(x-4\right)^2\left(x-16\right)^2}
                     \right] F_{32}
 \nonumber \\
 & &
 + \frac{J^2}{\left(2\pi i\right)^2} 
   \left[
          - \frac{\left(11x+16\right)}{x^2\left(x-16\right)} \frac{d \ln \omega_1}{dx}
          - \frac{\left(11x-14\right)}{x^2\left(x-4\right)\left(x-16\right)}
   \right]
 = 0,
 \\
\lefteqn{
 \frac{dF_{43}}{dx}
 + 2 \frac{dF_{32}}{dx}
 + \frac{3 J}{2\pi i} 
   \left[
           - \frac{2\left(x-10\right)}{\left(x-4\right)\left(x-16\right)} \frac{d \ln \omega_1}{dx}
           - \frac{2\left(x^3-30x^2+228x-640\right)}{x\left(x-4\right)^2\left(x-16\right)^2}
   \right]
 = 0.
} & & \nonumber
\eq
These differential equations follow from the requirement that terms of order $\eps^j$ with $j \neq 1$ are absent in the differential
equation for the master integrals:
Eq.~(\ref{eq_omega_L3}) removes the $\eps^{-2}$-term of $A_{4,2}$.
Eq.~(\ref{eq_dgl_Jacobian}) removes the $\eps^{0}$-term of $A_{4,4}$.
Eq.~(\ref{eq_omega_v1}) (together with the previous two equations) removes the $\eps^{-1}$-term of $A_{4,3}$.
$F_{32}$ removes the $\eps^{-1}$-term of $A_{4,2}$,
$F_{42}$ removes the $\eps^{0}$-term of $A_{4,2}$ and
$F_{43}$ removes the $\eps^{0}$-term of $A_{4,3}$.

With
\bq
 \frac{d^2 \ln \omega_1}{dx^2}
 & = &
 \frac{1}{\omega_1} \frac{d^2\omega_1}{dx^2}
 - \left( \frac{1}{\omega_1} \frac{d\omega_1}{dx} \right)^2
\eq
we may write eq.~(\ref{eq_omega_v1}) alternatively as
\bq
\label{eq_omega_v2}
 \frac{d^2 \ln \omega_1}{dx^2}
 + \frac{1}{2} \left( \frac{d \ln \omega_1}{dx} \right)^2
 + \frac{2\left(x^2-15x+32\right)}{x\left(x-4\right)\left(x-16\right)} \frac{d \ln \omega_1}{dx}
 + \frac{\left(x-8\right)}{2x\left(x-4\right)\left(x-16\right)}
 & = &
 0.
\eq
We have to check that eq.~(\ref{eq_omega_L3}) and eq.~(\ref{eq_omega_v1}) together have a solution.
The following lemma ensures this:
\begin{lemma}
Assume that $\omega_1$ satisfies eq.~(\ref{eq_omega_v1}). Then
\bq
 L_3^{(0)} \; \omega_1 & = & 0.
\eq
\end{lemma}
\begin{proof}
To prove this lemma, one verifies that the left-hand side of eq.~(\ref{eq_omega_L3}) simplifies to zero with the relation
eq.~(\ref{eq_omega_v1}).
\end{proof}
Thus it is sufficient to just consider the second-order non-linear inhomogeneous differential equation
eq.~(\ref{eq_omega_v1}) (or equivalently eq.~(\ref{eq_omega_v2})) and ignore eq.~(\ref{eq_omega_L3}).

It is not too difficult to solve eq.~(\ref{eq_omega_v2}):
We know that any solution of eq.~(\ref{eq_omega_v2})
is automatically a solution of eq.~(\ref{eq_omega_L3}). The solution space of eq.~(\ref{eq_omega_L3})
is given by
\bq
\label{linear_combination_omega_1}
 c_1 \psi_1^2 + c_2 \psi_1 \psi_2 + c_3 \psi_2^2
\eq
with unknown constants $c_1, c_2, c_3$.
\begin{lemma}
Let $\psi$ be a solution of eq.~(\ref{eq_psi_L2_psi}). Then $\psi^2$ is a solution of eq.~(\ref{eq_omega_v2}).
\end{lemma}
\begin{proof}
To prove this lemma, one verifies that the left-hand side of eq.~(\ref{eq_omega_v2}) simplifies to zero with the relation
eq.~(\ref{eq_psi_L2_psi}).
\end{proof}
In other words, $\omega_1$ needs to be a perfect square and of the form given in eq.~(\ref{linear_combination_omega_1}).
Thus,
\bq
 \psi_1^2, \;\; \psi_2^2
\eq
are solutions of eq.~(\ref{eq_omega_v2}), while $\psi_1 \psi_2$ is not.

From now on we set 
\bq
 \omega_1 & = & \psi_1^2.
\eq
Let us now turn to the Jacobian $J$.
It is easier to work with
\bq
 J^{-1} & = & \frac{d\tau}{dx}.
\eq
We may rewrite eq.~(\ref{eq_dgl_Jacobian}) as
\bq
\label{eq_dgl_Jacobian_v2}
 \frac{d \ln J^{-1}}{dx}
 + \frac{d \ln \omega_1}{dx}
 +\frac{2\left(x^2-15x+32\right)}{x\left(x-4\right)\left(x-16\right)}
 & = &
 0.
\eq
It is easily verified that
\bq
 \tau & = & \frac{\psi_2}{\psi_1}
\eq
solves eq.~(\ref{eq_dgl_Jacobian_v2}).
We then have
\bq
 J^{-1} & =& \frac{W}{\psi_1^2}.
\eq

% ------------------------------------------------------------------------------------------
\subsection{Step 2}
\label{sect:master_integrals_step_2}

With the information gathered so far
\bq
 \omega_1 \; = \; \psi_1^2, & & J \; = \; \frac{\psi_1^2}{W},
\eq
we may clean-up our ansatz.
The differential equations for $F_{42}$ and $F_{43}$ are of first order and easily solved.
Our ansatz reduces to
\bq
 I_1
 & = &
 \eps^3 I_{1110},
 \nonumber \\
 I_2
 & = &
 \eps^3 \frac{\pi^2}{\psi_1^2} I_{1111},
 \nonumber \\
 I_3
 & = &
 \frac{1}{2\pi i \eps} \frac{d}{d\tau} I_2
 + \left[ F_2 - \frac{\pi i \left(x-10\right)}{\left(x-4\right)\left(x-16\right)W} \left( \frac{\psi_1}{\pi} \right)^2 \right] I_2,
 \nonumber \\
 I_4
 & = &
 \frac{1}{2\pi i \eps} \frac{d}{d\tau} I_3
 + \left[ \frac{3}{2} F_2^2 + \frac{\pi^2 \left(x+8\right)^2\left(x^2-8x+64\right)}{8 x^2 \left(x-4\right)^2 \left(x-16\right)^2 W^2 } \left( \frac{\psi_1}{\pi} \right)^4 \right] I_2
 \nonumber \\
 & & 
 + \left[ - 2 F_2 - \frac{\pi i \left(x-10\right)}{\left(x-4\right)\left(x-16\right)W} \left( \frac{\psi_1}{\pi} \right)^2 \right] I_3
\eq
with one unknown function $F_2$, which satisfies
\bq
 \frac{d^2F_2}{dx^2} 
 + \left[\frac{2\left(x^2-15x+32\right)}{x\left(x-4\right)\left(x-16\right)} + 2 \left( \frac{d\ln \psi_1}{dx}\right) \right] \frac{dF_2}{dx}
 & = &
 \frac{\pi i\left(x-8\right)\left(x+8\right)^3}{x^2 \left(x-4\right)^3 \left(x-16\right)^3 W} \left( \frac{\psi_1}{\pi} \right)^{2}.
\eq
This is a linear inhomogeneous first-order differential equation for $F_2'$.
Noting that
\bq
 \frac{2\left(x^2-15x+32\right)}{x\left(x-4\right)\left(x-16\right)}
 + 2 \left(\frac{d \ln \psi_1}{dx} \right)
 & = &
 \frac{d \ln J}{dx}
\eq
one finds
\bq
 F_2' & = &
 J^{-1} \left[ C_2 
 + \int dx \; J 
 \frac{\pi i\left(x-8\right)\left(x+8\right)^3}{x^2 \left(x-4\right)^3 \left(x-16\right)^3 W} \left( \frac{\psi_1}{\pi} \right)^{2}
 \right]
\eq
and
\bq
 F_2
 & = &
 C_1
 +
 C_2 \tau
 +
 \left(2\pi i\right)^2
 \int d\tau \int d\tau \;  
 \frac{x\left(x-8\right)\left(x+8\right)^3}{864 \left(4-x\right)^{\frac{3}{2}} \left(16-x\right)^{\frac{3}{2}}} \left( \frac{\psi_1}{\pi} \right)^{6}.
\eq
We only need one specific solution. We may therefore set the integration constants to zero.
This yields
\bq
\label{def_F2}
 F_2
 & = &
 \left(2\pi i\right)^2
 \int\limits_{i\infty}^\tau d\tau_1 \int\limits_{i\infty}^{\tau_1} d\tau_2 \;  
 \frac{x\left(x-8\right)\left(x+8\right)^3}{864 \left(4-x\right)^{\frac{3}{2}} \left(16-x\right)^{\frac{3}{2}}} \left( \frac{\psi_1}{\pi} \right)^{6},
\eq
where the integrand is viewed as a function of $\tau_2$.
We denote the integrand by
\bq
\label{def_g6}
 g_6
 & = &
 \frac{x\left(x-8\right)\left(x+8\right)^3}{864 \left(4-x\right)^{\frac{3}{2}} \left(16-x\right)^{\frac{3}{2}}} \left( \frac{\psi_1}{\pi} \right)^{6}.
\eq
We note that the definition of $F_2$ depends on the integration path.
Later on we will expand the integrand in a $\qbar$-series. 
As long as we stay inside the region of convergence
of the $\qbar$-series we may suppress the dependence on the integration path.
We are only interested in this case and we therefore do not write the dependence on the integration path explicitly.
The dependence on the integration path is relevant as soon as we analytically continue $F_2$ beyond the region
of convergence of the $\qbar$-series.
Later on we will see that the radius of convergence of the $\qbar$-series is set by the threshold singularity at $x=16$,
corresponding to $\tau=1/2+i\sqrt{3}/6$.
As long as
\bq
 \mathrm{Im}\; \tau \; > \; \frac{\sqrt{3}}{6}
 & \;\;\; \mbox{or} \;\;\; &
 \left| \qbar \right| < e^{-\frac{\pi}{3}\sqrt{3}}
\eq
we do not have to worry about the path dependence.

% ------------------------------------------------------------------------------------------
\section{The periods}
\label{sect:periods}

In this section we consider the two solutions $\psi_1$ and $\psi_2$ of eq.~(\ref{eq_psi_L2_psi}).
It is well-known that these solutions are related to the two periods $\psi_1^{\mathrm{sunrise}}$ and $\psi_2^{\mathrm{sunrise}}$
of the sunrise integral.
In this section we review the construction.

The first step is the variable transformation
\bq
\label{def_trafo_x_y}
 x \; = \; - \frac{\left(y-1\right)\left(y-9\right)}{y},
 & &
 y \; = \; \frac{1}{2}\left[ 10-x-\sqrt{4-x}\sqrt{16-x}\right].
\eq
In expressing $y$ as a function of $x$, we made a choice for the sign of the square root.
With this choice the value $x=0$ is mapped to $y=1$.
This transformation rationalises the square root
\bq
\label{rationalise_sqrt}
 \sqrt{4-x}\sqrt{16-x}
 & = &
 - \frac{\left(y+3\right)\left(y-3\right)}{y}.
\eq
The sign is fixed by eq.~(\ref{def_trafo_x_y}).
This is most easily seen as follows: In the Euclidean region (i.e. $x\in]-\infty,0]$ and $y\in[0,1]$)
both sides of eq,~(\ref{rationalise_sqrt}) are positive.

In terms of the variable $y$, eq.~(\ref{eq_psi_L2_psi}) transforms into
\bq
 \left[ \frac{d^2}{dy^2}
 + \left( \frac{1}{y-1} + \frac{1}{y-9} \right) \frac{d}{dy}
 + \frac{y^2-2y+9}{4y^2\left(y-1\right)\left(y-9\right)} \right] \psi_i & = & 0.
\eq
If we now make the ansatz
\bq
 \psi_i & = & \sqrt{y} \; \psi_i^{\mathrm{sunrise}}
\eq
we find for $\psi_i^{\mathrm{sunrise}}$ the differential equation
\bq
 \left[ \frac{d^2}{dy^2}
 + \left( \frac{1}{y} + \frac{1}{y-1} + \frac{1}{y-9} \right) \frac{d}{dy}
 + \frac{y-3}{y\left(y-1\right)\left(y-9\right)} \right] \psi_i^{\mathrm{sunrise}} & = & 0.
\eq
This is the differential equation for the periods of the sunrise integral.
The solutions are well-known:
We consider an elliptic curve defined by the quartic polynomial
\bq
\label{def_generic_quartic_elliptic_curve}
 E
 & : &
 v^2 - \left(u-u_1\right) \left(u-u_2\right) \left(u-u_3\right) \left(u-u_4\right)
 \; = \; 0,
\eq
where the $u_j$ (with $j\in\{1,2,3,4\}$) denote the roots of the quartic polynomial.
The roots $u_j$ are given by
\bq
 u_1
 \; = \;
 -4,
 \;\;\;\;\;\;
 u_2
 \; = \;
 -\left(1+\sqrt{y}\right)^2,
 \;\;\;\;\;\;
 u_3
 \; = \;
 -\left(1-\sqrt{y}\right)^2,
 \;\;\;\;\;\;
 u_4
 \; = \;
 0.
\eq
We set
\bq
 U_1 \; = \; \left(u_3-u_2\right)\left(u_4-u_1\right),
 \;\;\;\;\;\;
 U_2 \; = \; \left(u_2-u_1\right)\left(u_4-u_3\right),
 \;\;\;\;\;\;
 U_3 \; = \; \left(u_3-u_1\right)\left(u_4-u_2\right).
\eq 
We define the modulus and the complementary modulus of the elliptic curve $E$ by
\bq
 k^2 
 \; = \; 
 \frac{U_1}{U_3},
 & &
 \bar{k}^2 
 \; = \;
 1 - k^2 
 \; = \;
 \frac{U_2}{U_3}.
\eq
Our standard choice for the periods and quasi-periods (for $x \in {\mathbb R}+i\delta$) is
\bq
\label{def_generic_periods}
 \left( \begin{array}{c}
 \psi_2^{\mathrm{sunrise}} \\
 \psi_1^{\mathrm{sunrise}} \\
 \end{array} \right)
 \; = \;
  \frac{4}{U_3^{\frac{1}{2}}}
 \gamma
 \left( \begin{array}{c}
 i K\left(\bar{k}\right) \\
 K\left(k\right) \\
 \end{array} \right),
 & &
 \left( \begin{array}{c}
 \phi_2^{\mathrm{sunrise}} \\
 \phi_1^{\mathrm{sunrise}} \\
 \end{array} \right)
 \; = \;
  \frac{4}{U_3^{\frac{1}{2}}}
 \gamma
 \left( \begin{array}{c}
 i E\left(\bar{k}\right)\\
 K\left(k\right) - E\left(k\right) \\
 \end{array} \right),
\eq
with
\bq
\label{def_gamma}
 \gamma
 & = &
 \left\{ \begin{array}{ll}
  \left(\begin{array}{cc}
   1 & 0 \\
   0 & 1 \\
  \end{array} \right) 
  & \mbox{if} \;\;\; x < 0 \;\;\; \mbox{or} \;\;\; 16+8\sqrt{3} < x,
  \\
  \left(\begin{array}{cc}
   1 & 0 \\
   2 & 1 \\
  \end{array} \right) 
  & \mbox{if} \;\;\; 0 < x \le 16+8\sqrt{3}.
  \\
 \end{array} \right. 
\eq
For $x \in [0,4]$ the modulus $k$ is real and $k>1$. Feynman's $i\delta$-prescription dictates that $K(k)$ and $E(k)$
are evaluated as $K(k+i\delta)$ and $E(k+i\delta)$.
The origin of the matrix $\gamma$ is explained below.
We have
\bq
\label{def_tau}
 \tau & = & \frac{\psi_2}{\psi_1}
 \; = \; 
 \frac{\psi_2^{\mathrm{sunrise}}}{\psi_1^{\mathrm{sunrise}}}.
\eq
We further set
\bq
 \qbar & = & e^{2 \pi i \tau}.
\eq
We may now determine the constant $c$ appearing in eq.~(\ref{eq_Wronskian}).
We have
\bq
 W^{\mathrm{sunrise}}
 & = &
 \psi_1^{\mathrm{sunrise}} \frac{d}{dy} \psi_2^{\mathrm{sunrise}}
 -
 \psi_2^{\mathrm{sunrise}} \frac{d}{dy} \psi_1^{\mathrm{sunrise}}
 \; = \;
 - \frac{6 \pi i}{y\left(1-y\right)\left(9-y\right)}
\eq
and
\bq
 W & = & 
 \frac{6 \pi i}{x \left(4-x\right)^{\frac{1}{2}} \left(16-x\right)^{\frac{1}{2}}},
\eq
hence $c=3$.

We may express $x$ and $y$ as a function of $\tau$ with the help of the following relations \cite{Verrill:1996,Maier:2006aa}:
\bq
\label{def_hauptmodul}
 x
 \; = \;
 -
 \left( \frac{\eta\left(\tau\right) \eta\left(3\tau\right)}
             {\eta\left(2\tau\right) \eta\left(6\tau\right)} \right)^6,
 & &
 y
 \; = \;
 9
 \frac{\eta\left(\tau\right)^4 \eta\left(6\tau\right)^8}
      {\eta\left(3\tau\right)^4 \eta\left(2\tau\right)^8}.
\eq
We have for example
\bq
 -\frac{1}{x} 
 & = &
 \qbar
 + 6 \qbar^2
 + 21 \qbar^3
 + 68 \qbar^4
 + 198 \qbar^5
 + {\mathcal O}\left(\qbar^6\right).
\eq
It is worth discussing the mapping in eq.~(\ref{def_trafo_x_y}) and the origin of the matrix $\gamma$ in eq.~(\ref{def_gamma}) in more detail:
We are interested in
\bq
 x & \in & {\mathbb R} + i \delta,
 \;\;\;\;\;\; \delta \; > \; 0,
\eq
where $i \delta$ denotes an infinitesimal imaginary part 
originating from Feynman's $i\delta$-prescription.
\begin{table}
\begin{center}
\begin{tabular}{|l|c|c|c|c|c|}
 \hline
 $x$ & $-\infty$ & $0$ & $4$ & $16$ & $\infty$ \\
 \hline
 $y$ & $0$ & $1$ & $3$ & $-3$ & $0$ \\
 \hline
 $\tau$ & $i\infty$ & $0$ & $\frac{1}{4}+\frac{i\sqrt{3}}{12}$ & $\frac{1}{2}+\frac{i\sqrt{3}}{6}$ & $i\infty$ \\
 \hline
 $\qbar$ & $0$ & $1$ & $ie^{-\frac{\pi}{6}\sqrt{3}}$ & $-e^{-\frac{\pi}{3}\sqrt{3}}$ & $0$ \\
 \hline
\end{tabular}
\end{center}
\caption{
Correspondence between special values in $x$-space, $y$-space, $\tau$-space and $\qbar$-space.
}
\label{table_x_y}
\end{table}
The point $x=-\infty+i\delta$ is mapped to $y=0$.
As we continue in $x$-space along the real line towards $x=0$, we traverse the interval $[0,1]$
in $y$-space.
Continuing in $x$-space to the point $x=4$, we reach the point $y=3$ in $y$-space.
The interval $[4,16]$ in $x$-space corresponds to a half-circle in the complex upper half-plane,
starting at $y=3$ and ending at $y=-3$.
Finally, the interval $[16,\infty[$ in $x$-space brings us back from $y=-3$ to $y=0$ in $y$-space.
\begin{figure}
\begin{center}
\includegraphics[scale=1.0]{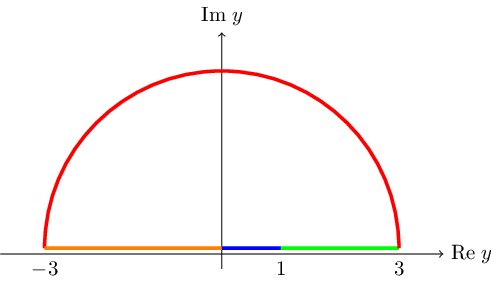}
\end{center}
\caption{
The path in $y$-space.
The interval $y \in [-3,0]$ corresponds to $x\in [16,\infty[$,
the interval $y \in [0,1]$ corresponds to $x\in ]-\infty,0]$,
the interval $y \in [1,3]$ corresponds to $x\in [0,4]$,
the half-circle corresponds to $x\in[4,16]$.
}
\label{fig_contour_y}
\end{figure}
This is summarised in table~\ref{table_x_y} and shown in fig.~\ref{fig_contour_y}.
We note that the cusps of $\Gamma_1(6)$ at $y=9$ and $y=\infty$ are never reached along this path.
The point $\qbar=0$ corresponds to $y=0$ and $x=\infty$.

We may now follow this path in $k^2$-space.
\begin{figure}
\begin{center}
\includegraphics[scale=1.0]{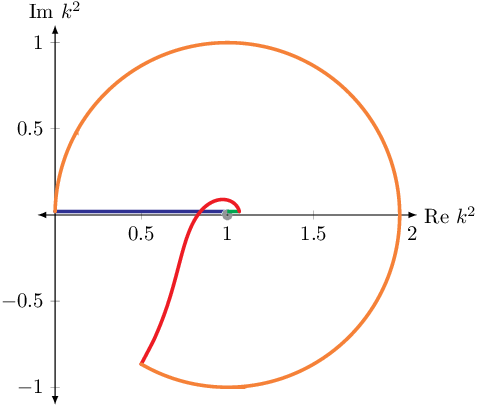}
\end{center}
\caption{
The path in $k^2$-space.
The path wraps in a small circle around the point $k^2=1$. 
}
\label{fig_contour_k2}
\end{figure}
The resulting contour is shown in fig.~\ref{fig_contour_k2}.
For $x=-\infty$ we start with $k^2=0$. As we increase $x$, we continue towards $k^2=1$.
At $x=0$ the path winds around $k^2=1$ in a small circle in the clockwise direction \cite{Bogner:2017vim,Honemann:2018mrb}.
In particular, for $k^2=1+\delta$ (with $\delta$ infinitesimal) the path crosses the branch cut of the
complete elliptic integrals.
The contour continues as shown in fig.~\ref{fig_contour_k2} and crosses the branch cut of the complete elliptic integrals
a second time at $k^2=2$, corresponding to $x=16+8\sqrt{3}$.
The periods $\psi_i^{\mathrm{sunrise}}$ and pseudo-periods $\phi_i^{\mathrm{sunrise}}$
are continuous functions of $x$.
The matrix $\gamma$ in eq.~(\ref{def_gamma}) compensates the discontinuity across the branch cut of the complete elliptic integrals and ensures that the periods $\psi_i^{\mathrm{sunrise}}$ and pseudo-periods $\phi_i^{\mathrm{sunrise}}$
are continuous.
With $k$ real and $k>1$ the discontinuities are given by
\bq
 K\left(k+i\delta\right) - K\left(k-i\delta\right) 
 & = &
 2 i K\left(\pm i \sqrt{k^2-1}\right),
 \nonumber \\
 \left[ K\left(k+i\delta\right) - E\left(k+i\delta\right) \right] - \left[ K\left(k-i\delta\right) - E\left(k-i\delta\right) \right] 
 & = &
 2 i E\left(\pm i \sqrt{k^2-1}\right).
\eq
We always have $\mathrm{Re}(k^2)>0$. This implies that the path for $\bar{k}^2=1-k^2$ never crosses
the branch cut of the complete elliptic integrals $K(\bar{k})$ and $E(\bar{k})$.
Hence, branch cut crossings do not occur for $\psi_2^{\mathrm{sunrise}}$ and $\phi_2^{\mathrm{sunrise}}$.

We may then look at the path in $\tau$-space and $\qbar$-space.
\begin{figure}
\begin{center}
\includegraphics[scale=1.0]{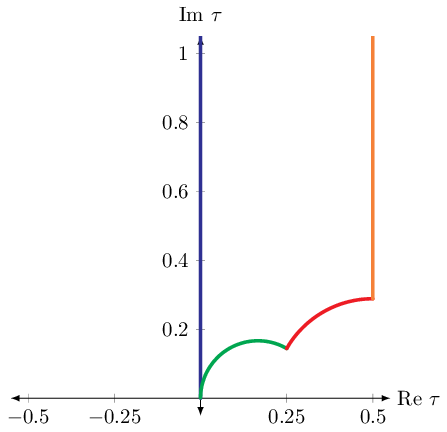}
\includegraphics[scale=1.0]{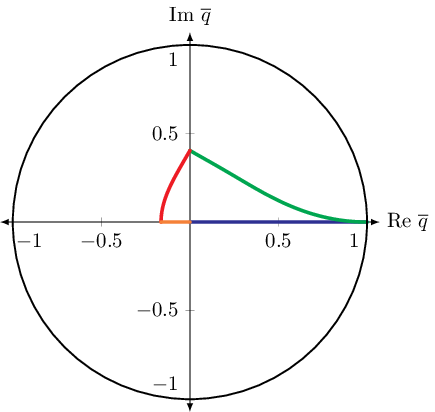}
\end{center}
\caption{
The path in $\tau$-space and in $\qbar$-space.
}
\label{fig_contour_tau_qbar}
\end{figure}
This is shown in fig.~\ref{fig_contour_tau_qbar}.
In the end we will expand in the variable $\qbar$ around $\qbar=0$. 
The radius of convergence is determined by the nearest pole, which in our case is located at
\bq
 \qbar = - e^{-\frac{\pi}{3}\sqrt{3}},
\eq
corresponding to $x=16$.
Hence, the expansion converges for
\bq
 x \in \left] -\infty, -2 \right[
 & \mbox{and} &
 x \in \left]16,\infty\right[.
\eq
The value $x=-2$ derives from the fact that
\bq
 \left| \qbar\left(x=-2\right)\right|
 \; = \;
 \left| \qbar\left(x=16\right)\right|
 \; = \;
 e^{-\frac{\pi}{3}\sqrt{3}}.
\eq

% ------------------------------------------------------------------------------------------
\section{Modular forms}
\label{sect:modular}

We introduce a basis $\{e_{1},e_{2}\}$ for the modular forms of modular weight $1$ 
for the Eisenstein subspace ${\mathcal E}_1(\Gamma_1(6))$:
\bq
 e_{1} \; = \; E_1\left(\tau;\chi_{1},\chi_{-3}\right),
 & &
 e_{2} \; = \; E_1\left(2\tau;\chi_{1},\chi_{-3}\right),
\eq
where 
\bq
 \chi_{1}\left(n\right) \; = \; \left( \frac{1}{n} \right)
 & \mbox{and} & 
 \chi_{-3}\left(n\right) \; = \; \left( \frac{-3}{n} \right)
\eq 
denote primitive Dirichlet characters with conductors $1$ and $3$, respectively.
Here, $(\frac{a}{n})$ denotes the Kronecker symbol in number theory.
The generalised Eisenstein series $E_k(\tau;\chi_a,\chi_b)$ are defined as follows \cite{Stein}, see also \cite{Weinzierl:2022eaz}:
Let $\chi_a$ and $\chi_b$ be primitive Dirichlet characters with conductors $d_a$ and $d_b$, respectively.
Then $E_k(\tau;\chi_a,\chi_b)$ is defined by 
\begin{align}
\label{def_Eisenstein_E}
E_{k}\left(\tau;\chi_a,\chi_b\right) &
 = 
 a_0 + \sum\limits_{n=1}^{\infty} \left( \sum\limits_{d|n} \chi_a(n/d) \cdot \chi_b(d) \cdot d^{k-1} \right) \qbar^{n}.
\end{align}
The normalisation is such that the coefficient of $\qbar$ is one.
The constant term $a_0$ is given by
\begin{align}
 a_0 &
 = 
 \begin{cases}
 -\frac{B_{k,\chi_b}}{2k}, \qquad &\text{if}\ d_a=1, \\
0, \qquad &\text{if}\ d_a>1.
\end{cases}
\end{align}
Note that the constant term $a_0$ depends on $\chi_a$ and $\chi_b$.
The generalised Bernoulli numbers
$B_{k,\chi_b}$ are defined by
\bq
\label{def_generalised_Bernoulli}
 \sum\limits_{n=1}^{d_b} \chi_b(n) \dfrac{xe^{nx}}{e^{d_b x}-1}
 & = &
 \sum\limits_{k=0}^{\infty} B_{k,\chi_b} \dfrac{x^k}{k!}. 
\eq
Instead of the basis $\{e_1,e_2\}$ we may use an alternative basis $\{b_0,b_1\}$ of ${\mathcal E}_1(\Gamma_1(6))$.
The latter basis is slightly more convenient for the conversion towards modular forms.
The basis $\{b_0,b_1\}$ is defined by
\bq
\label{def_b0_b1}
 b_0
 \; = \;
 \frac{\psi_1^{\mathrm{sunrise}}}{\pi}
 \; = \;
 2 \sqrt{3}
 \left( e_{1} + e_{2} \right),
 & &
 b_1
 \; = \;
 y \frac{\psi_1^{\mathrm{sunrise}}}{\pi} 
 \; = \; 
 6 \sqrt{3}
 \left( e_{1} - e_{2} \right).
\eq
An alternative representation of $b_0$ and $b_1$ is given as an eta-quotient:
\bq
 b_0 \; = \; \frac{2}{3} \sqrt{3} \frac{\eta\left(2\tau\right)^6 \eta\left(3\tau\right)}{\eta\left(\tau\right)^3 \eta\left(6\tau\right)^2},
 & &
 b_1 \; = \; 6 \sqrt{3} \frac{\eta\left(\tau\right)\eta\left(6\tau\right)^6}{\eta\left(2\tau\right)^2 \eta\left(3\tau\right)^3}.
\eq
Furthermore we may express $b_0$ and $b_1$ in terms of the coefficients $g^{(k)}(z,\tau)$
of the Kronecker function:
\bq
 b_0
 \; = \;
 \frac{1}{2\pi} \left[ g^{(1)}(\frac{1}{3},\tau) + g^{(1)}(\frac{1}{6},\tau) \right],
 & &
 b_1
 \; = \; 
 \frac{1}{2\pi} \left[ 9 g^{(1)}(\frac{1}{3},\tau) - 3 g^{(1)}(\frac{1}{6},\tau) \right].
\eq
In addition we introduce two meromorphic modular forms
\bq
 b_{3}
 \; = \;
 \frac{1}{\left(y-3\right)} \frac{\psi_1^{\mathrm{sunrise}}}{\pi},
 & &
 b_{-3}
 \; = \;
 \frac{1}{\left(y+3\right)} \frac{\psi_1^{\mathrm{sunrise}}}{\pi}.
\eq
We may express all integrands as polynomials in 
\bq
 b_0,
 \;\;\;
 b_1,
 \;\;\;
 b_{3},
 \;\;\;
 b_{-3}.
\eq
For example, the meromorphic modular form $g_6$ of eq.~(\ref{def_g6}) is expressed as
\bq
 g_6
 & = &
 \frac{1}{864} b_0 b_1^5 - \frac{11}{144} b_0^2 b_1^4 + \frac{107}{54} b_0^3 b_1^3 - \frac{421}{16} b_0^4 b_1^2 + \frac{6685}{32} b_0^5 b_1 - 1251 b_0^6 
 + 108 b_0^3 b_{3}^3
 \nonumber \\
 & &
 + 162 b_0^4 b_{3}^2 
 + 81 b_0^5 b_{3}
 + 6912 b_0^3 b_{-3}^3 - 10368 b_0^4 b_{-3}^2 + 6480 b_0^5 b_{-3}.
\eq
The $\qbar$-expansions of the four basic modular forms $b_0$, $b_1$, $b_3$ and $b_{-3}$ start as
\bq
 b_0
 & = &
 2 \sqrt{3} \left[
  \frac{1}{3} + \qbar + \qbar^2 + \qbar^3 + \qbar^4
 \right]
 + {\mathcal O}\left(\qbar^6\right),
 \nonumber \\
 b_1
 & = &
 6 \sqrt{3} \left[
  \qbar - \qbar^2 + \qbar^3 + \qbar^4
 \right]
 + {\mathcal O}\left(\qbar^6\right),
 \nonumber \\
 b_{3}
 & = &
 - \frac{4}{3} \sqrt{3} \left[
 \frac{1}{6} + \qbar + \frac{3}{2} \qbar^2 - 2 \qbar^3 - \frac{7}{2} \qbar^4 + 18 \qbar^5
 \right]
 + {\mathcal O}\left(\qbar^6\right),
 \nonumber \\
 b_{-3}
 & = &
 \frac{2}{9} \sqrt{3} \left[
  1 + 15 \qbar^2 - 72 \qbar^3 + 459 \qbar^4 - 2808 \qbar^5
 \right]
 + {\mathcal O}\left(\qbar^6\right).
\eq
We then have for example
\bq
 g_6
 & = &
 - 2 \qbar
 + 104 \qbar^2
 - 2286 \qbar^3
 + 30112 \qbar^4
 - 306300 \qbar^5
 + {\mathcal O}\left(\qbar^6\right).
\eq
The modular form $g_6$ is rather special: We have verified to very high order (${\mathcal O}(\qbar^{200})$) that all coefficients are integers and that the
coefficient of $\qbar^n$ is divisible by $n^2$.

% ------------------------------------------------------------------------------------------
\section{Iterated integrals of modular forms}
\label{sect:iterated_integrals}

Let $f_1(\tau)$, $f_2(\tau)$, ..., $f_n(\tau)$ be a set of (meromorphic) modular forms.
We define the $n$-fold iterated integral of these modular forms by
\begin{align}
I\left(f_1,f_2,...,f_n;\tau,\tau_0\right)
 & =
 \left(2 \pi i \right)^n
 \int\limits_{\tau_0}^{\tau} d\tau_1
 \int\limits_{\tau_0}^{\tau_1} d\tau_2
 ...
 \int\limits_{\tau_0}^{\tau_{n-1}} d\tau_n
 \;
 f_1\left(\tau_1\right)
 f_2\left(\tau_2\right)
 ...
 f_n\left(\tau_n\right).
\end{align}
With $\qbar=\exp(2\pi i \tau)$ we may equally well write
\begin{align}
I\left(f_1,f_2,...,f_n;\tau,\tau_0\right)
 & =
 \int\limits_{\qbar_0}^{\qbar} \frac{d\qbar_1}{\qbar_1}
 \int\limits_{\qbar_0}^{\qbar_1} \frac{d\qbar_2}{\qbar_2}
 ...
 \int\limits_{\qbar_0}^{\qbar_{n-1}} \frac{d\qbar_n}{\qbar_n}
 \;
 f_1\left(\tau_1\right)
 f_2\left(\tau_2\right)
 ...
 f_n\left(\tau_n\right),
 \qquad \tau_j = \frac{1}{2\pi i} \ln \qbar_j.
\end{align}
Our standard choice for the base point $\tau_0$ will be $\tau_0 = i \infty$, corresponding to $q_0=0$.
If $f_n(\tau)$ does not vanish at the cusp $\tau=i\infty$ we employ the standard ``trailing zero'' or ``tangential base point''
regularisation \cite{Brown:2014aa,Adams:2017ejb,Walden:2020odh}:
We first take $\qbar_0$ to have a small non-zero value.
The integration will produce terms with $\ln(\qbar_0)$. Let $R$ be the operator, which removes all $\ln(\qbar_0)$-terms.
After these terms have been removed, we may take the limit $\qbar_0\rightarrow 0$.
With this regularisation we set
\begin{align}
I\left(f_1,f_2,...,f_n;\tau\right)
 & =
 \lim\limits_{\qbar_0\rightarrow 0}
 R \left[
 \int\limits_{\qbar_0}^{\qbar} \frac{d\qbar_1}{\qbar_1}
 \int\limits_{\qbar_0}^{\qbar_1} \frac{d\qbar_2}{\qbar_2}
 ...
 \int\limits_{\qbar_0}^{\qbar_{n-1}} \frac{d\qbar_n}{\qbar_n}
 \;
 f_1\left(\tau_1\right)
 f_2\left(\tau_2\right)
 ...
 f_n\left(\tau_n\right)
 \right].
\end{align}
In this notation the expression $F_2$ from eq.~(\ref{def_F2}) becomes
\bq
 F_2
 & = & 
 I\left(1,g_6;\tau\right).
\eq
Furthermore it will be convenient to introduce a short-hand notation for repeated letters.
We use the notation
\begin{align}
 \left\{f_i\right\}^j & =
 \underbrace{ f_i, f_i, ..., f_i}_{j}
\end{align}
to denote a sequence of $j$ letters $f_i$.
For iterated integrals we have the shuffle product, for example
\bq
 I\left(f_1,f_2;\tau\right)  I\left(f_3;\tau\right)
 & = &
 I\left(f_1,f_2,f_3;\tau\right)
 + I\left(f_1,f_3,f_2;\tau\right)
 + I\left(f_3,f_1,f_2;\tau\right).
\eq
Using the antipode in the shuffle algebra \cite{Weinzierl:2022eaz} one easily shows that
\bq
\label{antipode_relation}
 I\left(f_k,\left\{1\right\}^j;\tau\right)
 & = &
 \sum\limits_{i=0}^j 
  \left(-1\right)^i I\left(\left\{1\right\}^{j-i};\tau\right) I\left(\left\{1\right\}^i,f_k;\tau\right),
 \nonumber \\
 I\left(\left\{1\right\}^j,f_k;\tau\right)
 & = &
 \sum\limits_{i=0}^j 
  \left(-1\right)^i I\left(\left\{1\right\}^{j-i};\tau\right) I\left(f_k,\left\{1\right\}^i;\tau\right).
\eq
These relations will be useful when we work out the modular transformation properties of $F_2$.

% ------------------------------------------------------------------------------------------
\section{The differential equation}
\label{sect:dgl}

For the basis $I$ the differential equation is in $\eps$-form
\bq
 d I & = & \eps A I,
\eq
with
\bq
 A & = &
 2 \pi i 
 \left( \begin{array}{cccc}
 0 & 0 & 0 & 0 \\
 0 & -f_{2,a}-f_{2,b} & 1 & 0 \\
 0 & f_{4,b} & -f_{2,a}+2f_{2,b} & 1 \\
 f_{4,a} & f_{6} & f_{4,b} & -f_{2,a}-f_{2,b} \\
 \end{array} \right) d\tau.
\eq
This is a differential equation with an alphabet consisting of six letters
\bq
 {\mathcal A}
 & = &
 \left\{
  1, f_{2,a}, f_{2,b}, f_{4,a}, f_{4,b}, f_6
 \right\}.
\eq
Below we give expressions for all non-trivial letters in $x$-space, $y$-space and $\tau$-space.

The differential one-forms $f_{2,a}$ and $f_{2,b}$ are of modular weight $2$.
The meromorphic modular form $f_{2,a}$ is just the translation of a dlog-form to the variable $\tau$:
\bq
f_{2,a} \cdot 2 \pi i d\tau & = & d\ln\left(x-4\right) + d\ln\left(x-16\right).
\eq
We have
\bq
 f_{2,a}
 & = &
 \left( \frac{1}{x-4} + \frac{1}{x-16} \right) \frac{\psi_1^2}{2\pi i W}
 \nonumber \\
 & = &
 \left[ \frac{1}{6} y^2 - \frac{5}{3} y + \frac{9}{2} - \frac{6}{y-3} - \frac{24}{y+3} \right] \left( \frac{\psi_1^{\mathrm{sunrise}}}{\pi} \right)^2
 \nonumber \\
 & = &
 \frac{1}{6} b_1^2 - \frac{5}{3} b_0 b_1 + \frac{9}{2} b_0^2 - 6 b_0 b_{3} - 24 b_0 b_{-3}.
\eq
The differential one-form $f_{2,b}$ is given by
\bq
 f_{2,b}
 & = &
 F_2,
\eq
with 
\bq
 F_2
 & = & 
 I\left(1,g_6;\tau\right)
\eq
and
\bq
\label{def_g6_v2}
 g_6
 & = &
 \frac{x\left(x-8\right)\left(x+8\right)^3}{864 \left(4-x\right)^{\frac{3}{2}} \left(16-x\right)^{\frac{3}{2}}} \left( \frac{\psi_1}{\pi} \right)^{6}
 \nonumber \\
 & = &
 \frac{y\left(y-1\right)\left(y-9\right)\left(y^2-2y+9\right)\left(y^2-18y+9\right)^3}{864\left(y-3\right)^3\left(y+3\right)^3}\left( \frac{\psi_1^{\mathrm{sunrise}}}{\pi} \right)^6
 \nonumber \\
 & = &
 \frac{1}{864} b_0 b_1^5 - \frac{11}{144} b_0^2 b_1^4 + \frac{107}{54} b_0^3 b_1^3 - \frac{421}{16} b_0^4 b_1^2 + \frac{6685}{32} b_0^5 b_1 - 1251 b_0^6 
 + 108 b_0^3 b_{3}^3
 \nonumber \\
 & &
 + 162 b_0^4 b_{3}^2 
 + 81 b_0^5 b_{3}
 + 6912 b_0^3 b_{-3}^3 - 10368 b_0^4 b_{-3}^2 + 6480 b_0^5 b_{-3}.
\eq
$f_{4,a}$ and $f_{4,b}$ are of modular weight $4$.
The differential one-form $f_{4,a}$ is a (holomorphic) modular form and given by
\bq
 f_{4,a}
 & = &
 \frac{1}{72} x \left(4-x\right)^{\frac{1}{2}} \left(16-x\right)^{\frac{1}{2}} \left( \frac{\psi_1}{\pi}\right)^4
 \nonumber \\
 & = &
 \frac{1}{72} \left(y-1\right)\left(y-9\right) \left(y-3\right)\left(y+3\right) \left( \frac{\psi_1^{\mathrm{sunrise}}}{\pi} \right)^4
 \nonumber \\
 & = &
 \frac{1}{72} b_1^4 - \frac{5}{36} b_0 b_1^3 + \frac{5}{4} b_0^3 b_1 - \frac{9}{8} b_0^4.
\eq
The differential one-form $f_{4,b}$ is given by
\bq
 f_{4,b}
 & = &
 \frac{\left(x+8\right)^2\left(x^2-8x+64\right)}{288 \left(x-4\right)\left(x-16\right)} \left( \frac{\psi_1}{\pi}\right)^4
 - \frac{3}{2} F_2^2
 \nonumber \\
 & = &
 \frac{\left(y^2-18y+9\right)^2\left(y^4-12y^3+102y^2-108y+81\right)}{288 \left(y-3\right)^2\left(y+3\right)^2} \left( \frac{\psi_1^{\mathrm{sunrise}}}{\pi} \right)^4
 - \frac{3}{2} F_2^2
 \nonumber \\
 & = &
 \frac{1}{288} b_1^4 - \frac{1}{6} b_0 b_1^3 + \frac{149}{48} b_0^2 b_1^2 - \frac{63}{2} b_0^3 b_1 + \frac{6537}{32} b_0^4
 + 54 b_0^2 b_{3}^2 + 54 b_0^3 b_{3}
 \nonumber \\
 & &
 + 864 b_0^2 b_{-3}^2 - 864 b_0^3 b_{-3}
 - \frac{3}{2} F_2^2.
\eq
Finally, at modular weight $6$ we have
\bq
 f_{6}
 & = &
 \frac{x\left(x-8\right)\left(x+8\right)^3}{216 \left(4-x\right)^{\frac{3}{2}} \left(16-x\right)^{\frac{3}{2}}} \left( \frac{\psi_1}{\pi}\right)^6
 - \frac{\left(x+8\right)^2\left(x^2-8x+64\right)}{144 \left(x-4\right)\left(x-16\right)} \left( \frac{\psi_1}{\pi} \right)^4 F_2
 + F_2^3
 \nonumber \\
 & = &
 \frac{y \left(y-1\right)\left(y-9\right)\left(y^2-2y+9\right)\left(y^2-18y+9\right)^3}{216\left(y-3\right)^3\left(y+3\right)^3}\left( \frac{\psi_1^{\mathrm{sunrise}}}{\pi} \right)^6
 \nonumber \\
 & &
 - \frac{\left(y^2-18y+9\right)^2\left(y^4-12y^3+102y^2-108y+81\right)}{144 \left(y-3\right)^2\left(y+3\right)^2} \left( \frac{\psi_1^{\mathrm{sunrise}}}{\pi} \right)^4 F_2
 + F_2^3
 \nonumber \\
 & = &
 \frac{1}{216} b_0 b_1^5 - \frac{11}{36} b_0^2 b_1^4 + \frac{214}{27} b_0^3 b_1^3 - \frac{421}{4} b_0^4 b_1^2 + \frac{6685}{8} b_0^5 b_1 - 5004 b_0^6 
 + 432 b_0^3 b_{3}^3
 \nonumber \\
 & &
 + 648 b_0^4 b_{3}^2 
 + 324 b_0^5 b_{3}
 + 27648 b_0^3 b_{-3}^3 - 41472 b_0^4 b_{-3}^2 + 25920 b_0^5 b_{-3}
 - \left[  \frac{1}{144} b_1^4 - \frac{1}{3} b_0 b_1^3
 \right.
 \nonumber \\
 & &
 \left.
 + \frac{149}{24} b_0^2 b_1^2 - 63 b_0^3 b_1 + \frac{6537}{16} b_0^4
 + 108 b_0^2 b_{3}^2 + 108 b_0^3 b_{3}
 + 1728 b_0^2 b_{-3}^2 - 1728 b_0^3 b_{-3}
 \right] F_2
 \nonumber \\
 & &
 + F_2^3.
\eq
With the help of eq.~(\ref{def_hauptmodul}) and eq.~(\ref{def_b0_b1}) the
$\qbar$-expansions of $f_{2,a}$, $f_{2,b}$, $f_{4,a}$, $f_{4,b}$ and $f_6$ are readily obtained.
The first few terms read
\bq
 f_{2,a}
 & = &
 - 2 + 8 \qbar - 44 \qbar^2 + 440 \qbar^3 - 2956 \qbar^4 + 17328 \qbar^5 
 + {\mathcal O}\left(\qbar^6\right),
 \nonumber \\
 f_{2,b}
 & = &
 - 2 \qbar + 26 \qbar^2 - 254 \qbar^3 + 1882 \qbar^4 - 12252 \qbar^5
 + {\mathcal O}\left(\qbar^6\right),
 \nonumber \\
 f_{4,a}
 & = &
 -2 - 4 \qbar + 28 \qbar^2 - 76 \qbar^3 + 284 \qbar^4 - 504 \qbar^5
 + {\mathcal O}\left(\qbar^6\right),
 \nonumber \\
 f_{4,b}
 & = &
 \frac{1}{2} - 18 \qbar + 282 \qbar^2 -3150 \qbar^3 + 28314 \qbar^4 - 200268 \qbar^5
 + {\mathcal O}\left(\qbar^6\right),
 \nonumber \\
 f_{6}
 & = &
 - 6 \qbar + 318 \qbar^2 - 6810 \qbar^3 + 81534 \qbar^4 - 710676 \qbar^5
 + {\mathcal O}\left(\qbar^6\right).
\eq
We have checked to very high order (${\mathcal O}(\qbar^{200})$) that with the exception of the constant term of $f_{4,b}$ 
all coefficients are integers.

Let us now look at the poles at $x=4$ and $x=16$ in $\tau$-space.
The Jacobian for the transformation from $\tau$-space to $y$-space is smooth in a neighbourhood of these points,
so we may discuss the poles in $y$-space.
From eq.~(\ref{def_g6_v2}) we see that $g_6$ has a triple pole at $y=3$ (corresponding to $x=4$) and at $y=-3$ (corresponding to $x=16$).
However, $f_{2,b}$ has only a simple pole at $y=\pm3$.
For $f_{4,b}$ and $f_6$, the triple and double poles appearing in the individual terms of their definition cancel in the sum,
leaving $f_{4,b}$ and $f_6$ with a simple pole at $y=\pm3$.
In summary it follows that all entries of the differential equation have at most a simple pole at $y=\pm 3$.

% -----------------------------------------------------------------------------

\section{Modular transformations}
\label{sect:modular_transformation}

In this section we discuss the behaviour of $f_{2,a}$, $f_{2,b}$, $f_{4,a}$, $f_{4,b}$ and $f_{6}$ under 
modular transformations.
We start with a few definitions.
Let 
\bq
 \gamma 
 & = &
 \left(\begin{array}{cc}
       a & b \\
       c & d \\
  \end{array} \right) 
 \; \in \; \Gamma.
\eq 
We consider the transformation
\bq
 \tau' 
 \; = \; 
 \gamma\left(\tau\right)
 \; = \;
 \frac{a\tau+b}{c\tau+d},
 & &
 d\tau'
 \; = \;
 \frac{d\tau}{\left(c\tau+d\right)^2}.
\eq
The point $\tau'=\frac{a}{c}$ corresponds to $\tau=i\infty$.

Let $f : {\mathbb H} \rightarrow \hat{{\mathbb C}}$ be a function from the complex upper half-plane to 
$\hat{{\mathbb C}} = {\mathbb C} \cup \{\infty\}$. 
We define
\bq
 (f \slashoperator{\gamma}{k})(\tau) & = & (c\tau+d)^{-k} \cdot f(\gamma(\tau)).
\eq
We say that $f$ transforms as a modular form of modular weight $k$ (or is weakly modular) if
\bq
\label{def_modular_trafo}
 (f \slashoperator{\gamma}{k})(\tau) & = & f(\tau).
\eq
We say that $f$ transforms as a quasi-modular form of modular weight $k$ and depth $p$ if there are $f_1, \dots, f_p$
such that
\bq
\label{def_quasi_modular_trafo}
 (f \slashoperator{\gamma}{k})(\tau)
 & = &
 f(\tau)
 + \sum\limits_{j=1}^p \left( \frac{c}{c\tau+d} \right)^j f_j(\tau).
\eq
We need one more generalisation:
We say that $f$ transforms as ``quasi-Eichler'' of modular weight $k$ and depth $p$ if
\bq
\label{def_quasi_Eichler}
 (f \slashoperator{\gamma}{k})(\tau)
 & = &
 f(\tau)
 + \sum\limits_{j=1}^p \left( \frac{c}{c\tau+d} \right)^j f_j(\tau)
 + \frac{P_\gamma(\tau)}{\left(c\tau+d\right)^p},
\eq
where $P_\gamma(\tau)$ is a polynomial in $\tau$
of degree at most $(p-k)$.
Note that the $f_j$'s are independent of $\gamma$, while the polynomial $P_\gamma$ may depend on $\gamma$.
A special case is an Eichler integral, which corresponds to the case $p=0$:
$f$ transforms as an Eichler integral of modular weight $k$ (with $k<0$)
if there is a polynomial $P_\gamma(\tau)$ of degree
at most $(-k)$ such that
\bq
 (f \slashoperator{\gamma}{k})(\tau)
 & = &
 f(\tau)
 + P_\gamma(\tau).
\eq
We now specialise to $\Gamma=\Gamma_1(6)$.
Generators of $\Gamma_1(6)$ are
\bq
 \left\{
 \left(\begin{array}{rr}
  1 & 1 \\
  0 & 1 \\
 \end{array} \right),
 \left(\begin{array}{rr}
  -5 & 1 \\
  -6 & 1 \\
 \end{array} \right),
 \left(\begin{array}{rr}
  7 & -3 \\
  12 & -5 \\
 \end{array} \right)
 \right\}.
\eq
$f_{4,a}$ is a (holomorphic) modular form for $\Gamma_1(6)$, $f_{2,a}$ is a meromorphic modular form for $\Gamma_1(6)$.
Both transform for $\gamma \in \Gamma_1(6)$ according to eq.~(\ref{def_modular_trafo}).

We work out the behaviour of $f_{2,b}$ under modular transformations.
We start from 
\bq
 f_{2,b}\left(\tau'\right) & = & I\left(1,g_6;i \infty, \tau'\right).
\eq
The path composition formula for iterated integrals gives us
\bq
 f_{2,b}\left(\tau'\right)
 & = &
 I\left(1,g_6;\frac{a}{c}, \tau'\right) 
 + I\left(g_6;i \infty, \frac{a}{c}\right) I\left(1;\frac{a}{c}, \tau'\right)     
 + I\left(1,g_6;i \infty, \frac{a}{c}\right).
\eq
Let us emphasise that we require the combined path in $\tau'$-space, consisting of the path from
$i \infty$ to $a/c$ followed by the path from $a/c$ to $\tau'$ followed by the reverse path from $\tau'$ to $i\infty$
to be homotopy equivalent to the zero path (i.e. no poles are encircled and no branch cuts are crossed).
We define two constants
\bq
 C_{6}
 \; = \;
 I\left(g_6;i \infty, \frac{a}{c}\right),
 & & 
 C_{1,6}
 \; = \; 
 I\left(1,g_6;i \infty, \frac{a}{c}\right).
\eq
Furthermore we have
\bq
 I\left(1;\frac{a}{c}, \tau'\right) 
 & = &
 2 \pi i \left(\tau'-\frac{a}{c} \right)
 \; = \;
 - \frac{2\pi i}{c} \left(c\tau+d\right)^{-1}.
\eq
A leading one will introduce a negative power of the automorphic factor $(c\tau+d)$ in the integrand.
This can be avoided by first using eq.~(\ref{antipode_relation}) to convert any leading $1$'s to trailing $1$'s, 
using then the modular transformation law for the integrands and in the end converting back to leading $1$'s, again with the
help of eq.~(\ref{antipode_relation}).
We find
\bq
 (f_{2,b} \slashoperator{\gamma}{2})(\tau)
 & = &
 f_{2,b}(\tau)
 - 6 \frac{c}{c\tau+d} \frac{1}{2\pi i} I\left(1,1,g_6;\tau\right)
 + 18 \left(\frac{c}{c\tau+d}\right)^2 \frac{1}{\left(2\pi i\right)^2} I\left(1,1,1,g_6;\tau\right)
 \nonumber \\
 & &
 - 24 \left(\frac{c}{c\tau+d}\right)^3 \frac{1}{\left(2\pi i\right)^3} I\left(1,1,1,1,g_6;\tau\right)
 + \frac{C_{1,6}}{\left(c\tau+d\right)^2}
 - \frac{2 \pi i C_{6}}{c \left(c\tau+d\right)^3}.
\eq
We see that $f_{2,b}$ transforms as quasi-Eichler of modular weight $2$ and depth $3$.

Once the transformation law for $F_2=f_{2,b}$ under modular transformations $\gamma \in \Gamma_1(6)$ is known,
the transformation laws for $f_{4,b}$ and $f_6$ follow from the definition of these quantities.
They are given as polynomials of building blocks, whose transformation is known.

% -----------------------------------------------------------------------------

\section{Analytical results}
\label{sect:results}

We write for the $\eps$-expansion of the master integrals
\bq
 I_j & = & \sum\limits_{k=0}^\infty I_j^{(k)} \eps^k.
\eq
Up to order $\eps^4$ the results are still relatively compact and we list them below.
Results up to order $\eps^6$ can be found in an ancillary file attached to this article.
The master integral $I_1$ is very simple and given by
\bq
 I_1
 & = &
 1
 + \frac{3}{2} \zeta_2 \eps^2
 - \zeta_3 \eps^3
 + \frac{57}{16} \zeta_4 \eps^4
 + {\mathcal O}\left(\eps^5\right).
\eq
The interesting master integrals are $I_2$, $I_3$ and $I_4$.
The master integral $I_2$ starts at order $\eps^3$.
Its terms of order $\eps^3$ and $\eps^4$ are given by
\bq
 I_2^{(3)}
 & = &
 \frac{4}{3} \zeta_3 + I\left(1,1,f_{4,a};\tau\right),
 \nonumber \\
 I_2^{(4)}
 & = &
 2 \zeta_4
          + \frac{4}{3} \zeta_3 \left[ \frac{11}{2} \ln\left(\qbar\right) - I\left(f_{2,a};\tau\right) - I\left(f_{2,b};\tau\right) \right]
          + \zeta_2 \ln^2\left(\qbar\right)
          - I\left(1,1,f_{2,a},f_{4,a};\tau\right)
 \nonumber \\ 
 & &
          - I\left(1,f_{2,a},1,f_{4,a};\tau\right) - I\left(f_{2,a},1,1,f_{4,a};\tau\right)
          - I\left(1,1,f_{2,b},f_{4,a};\tau\right)
 \nonumber \\ 
 & &
          + 2 I\left(1,f_{2,b},1,f_{4,a};\tau\right)
          - I\left(f_{2,b},1,1,f_{4,a};\tau\right).
\eq
The term $I_2^{(3)}$ agrees with the result of Bloch, Kerr and Vanhove \cite{Bloch:2014qca}.
In our notation their result reads
\bq
 I_2^{(3)}
 & = &
 \frac{4}{3} \zeta_3 - \frac{1}{3}\ln^3\left(\qbar\right) 
        - 4 \sum\limits_{n=1}^\infty \frac{\chi\left(n\right)}{n^3} \frac{\qbar^n}{1-\qbar^n},
\eq
where $\chi$ is a character of modulus six (i.e. $\chi(n+6)=\chi(n)$), taking the values
\bq
\begin{array}{c|rrrrrr}
 n & 0 & 1 & 2 & 3 & 4 & 5 \\
 \hline
 \chi & 120 & 1 & -15 & -8 & -15 & 1 \\
\end{array}.
\eq
The master integral $I_3$ starts at order $\eps^2$. The non-zero terms up to order $\eps^4$ read
\bq
 I_3^{(2)}
 & = &
 I\left(1,f_{4,a};\tau\right),
 \nonumber \\
 I_3^{(3)}
 & = &
 \frac{22}{3} \zeta_3
          + 2 \zeta_2 \ln\left(\qbar\right)
          - I\left(1,f_{2,a},f_{4,a};\tau\right) - I\left(f_{2,a},1,f_{4,a};\tau\right)
          - I\left(1,f_{2,b},f_{4,a};\tau\right)
 \nonumber \\
 & &
          + 2 I\left(f_{2,b},1,f_{4,a};\tau\right),
 \nonumber \\
 I_3^{(4)}
 & = &
          \frac{17}{2} \zeta_4
          + \frac{4}{3} \zeta_3 \left[ 8 \ln\left(\qbar\right) - \frac{11}{2} I\left(f_{2,a};\tau\right)
                                       + 11 I\left(f_{2,b};\tau\right) + I\left(f_{4,b};\tau\right) \right]
 \nonumber \\
 & &
          - 2 \zeta_2 \left[
                             I\left(1,f_{2,a};\tau\right) + I\left(f_{2,a},1;\tau\right) 
                           + I\left(1,f_{2,b};\tau\right) - 2 I\left(f_{2,b},1;\tau\right) 
                           - \frac{3}{4} I\left(1,f_{4,a};\tau\right)
                      \right]
 \nonumber \\
 & &
          + I\left(1,f_{2,a},f_{2,a},f_{4,a};\tau\right)
          + I\left(f_{2,a},1,f_{2,a},f_{4,a};\tau\right)
          + I\left(f_{2,a},f_{2,a},1,f_{4,a};\tau\right)
 \nonumber \\
 & &
          + I\left(1,f_{2,a},f_{2,b},f_{4,a};\tau\right)
          + I\left(1,f_{2,b},f_{2,a},f_{4,a};\tau\right)
          - 2 I\left(f_{2,b},1,f_{2,a},f_{4,a};\tau\right)
 \nonumber \\
 & &
          + I\left(f_{2,a},1,f_{2,b},f_{4,a};\tau\right)
          - 2 I\left(f_{2,a},f_{2,b},1,f_{4,a};\tau\right)
          - 2 I\left(f_{2,b},f_{2,a},1,f_{4,a};\tau\right)
 \nonumber \\
 & &
          + I\left(1,f_{2,b},f_{2,b},f_{4,a};\tau\right)
          - 2 I\left(f_{2,b},1,f_{2,b},f_{4,a};\tau\right)
          + 4 I\left(f_{2,b},f_{2,b},1,f_{4,a};\tau\right)
 \nonumber \\
 & &
          + I\left(1,f_{4,b},1,f_{4,a};\tau\right)
          + I\left(f_{4,b},1,1,f_{4,a};\tau\right).
\eq
The master integral $I_4$ starts at order $\eps$. The non-zero terms up to order $\eps^4$ read
\bq
 I_4^{(1)}
 & = &
 I\left(f_{4,a};\tau\right),
 \nonumber \\
 I_4^{(2)}
 & = &
 2 \zeta_2 
          - I\left(f_{2,a},f_{4,a};\tau\right)
          - I\left(f_{2,b},f_{4,a};\tau\right),
 \nonumber \\
 I_4^{(3)}
 & = &
 \frac{32}{3} \zeta_3 
          - 2 \zeta_2 \left[
                             I\left(f_{2,a};\tau\right)
                           + I\left(f_{2,b};\tau\right)
                           - \frac{3}{4} I\left(f_{4,a};\tau\right)
                      \right]
          + I\left(f_{2,a},f_{2,a},f_{4,a};\tau\right)
 \nonumber \\
 & &
          + I\left(f_{2,a},f_{2,b},f_{4,a};\tau\right)
          + I\left(f_{2,b},f_{2,a},f_{4,a};\tau\right)
          + I\left(f_{2,b},f_{2,b},f_{4,a};\tau\right)
          + I\left(f_{4,b},1,f_{4,a};\tau\right),
 \nonumber \\
 I_4^{(4)}
 & = &
          \frac{39}{2} \zeta_4
          - \frac{4}{3} \zeta_3 \left[
                                       8 I\left(f_{2,a};\tau\right)
                                       + 8 I\left(f_{2,b};\tau\right)
                                       + \frac{3}{4} I\left(f_{4,a};\tau\right)
                                       - \frac{11}{2} I\left(f_{4,b};\tau\right)
                                       - I\left(f_{6};\tau\right)
                                \right]
 \nonumber \\
 & &
          + 2 \zeta_2 \left[
                             I\left(f_{2,a},f_{2,a};\tau\right)
                             + I\left(f_{2,a},f_{2,b};\tau\right)
                             + I\left(f_{2,b},f_{2,a};\tau\right)
                             + I\left(f_{2,b},f_{2,b};\tau\right)
 \right. \nonumber \\
 & & \left.
                             - \frac{3}{4} I\left(f_{2,a},f_{4,a};\tau\right)
                             - \frac{3}{4} I\left(f_{2,b},f_{4,a};\tau\right)
                             + I\left(f_{4,b},1;\tau\right)
                       \right]
     - I\left(f_{2,a},f_{2,a},f_{2,a},f_{4,a};\tau\right)
 \nonumber \\
 & &
     - I\left(f_{2,a},f_{2,a},f_{2,b},f_{4,a};\tau\right)
     - I\left(f_{2,a},f_{2,b},f_{2,a},f_{4,a};\tau\right)
     - I\left(f_{2,b},f_{2,a},f_{2,a},f_{4,a};\tau\right)
 \nonumber \\
 & &
     - I\left(f_{2,a},f_{2,b},f_{2,b},f_{4,a};\tau\right)
     - I\left(f_{2,b},f_{2,a},f_{2,b},f_{4,a};\tau\right)
     - I\left(f_{2,b},f_{2,b},f_{2,a},f_{4,a};\tau\right)
 \nonumber \\
 & &
     - I\left(f_{2,b},f_{2,b},f_{2,b},f_{4,a};\tau\right)
     - I\left(f_{2,a},f_{4,b},1,f_{4,a};\tau\right)
     - I\left(f_{2,b},f_{4,b},1,f_{4,a};\tau\right)
 \nonumber \\
 & &
     - I\left(f_{4,b},1,f_{2,a},f_{4,a};\tau\right)
     - I\left(f_{4,b},1,f_{2,b},f_{4,a};\tau\right)
     - I\left(f_{4,b},f_{2,a},1,f_{4,a};\tau\right)
 \nonumber \\
 & &
     + 2 I\left(f_{4,b},f_{2,b},1,f_{4,a};\tau\right)
     + I\left(f_{6},1,1,f_{4,a};\tau\right).
\eq
As $f_{2,b}$ is itself an iterated integral
\bq
 f_{2,b}(\tau) & = & I\left(1,g_6;\tau\right),
\eq
we may in principle eliminate $f_{2,b}$ in favour of $g_6$.
For example
\bq
 I\left(1,f_{2,b},f_{4,a};\tau\right)
 & = &
   I\left(1,1,1,g_6,f_{4,a};\tau\right)
 + I\left(1,1,1,f_{4,a},g_6;\tau\right)
 + I\left(1,1,f_{4,a},1,g_6;\tau\right).
 \;\;\;
\eq
In this way we obtain only iterated integrals of meromorphic modular forms, confirming the result of \cite{Broedel:2021zij}.
However there is a price to pay: Doing so, we introduce integrands with higher poles and we spoil the uniform depth property.

% -----------------------------------------------------------------------------

\section{Numerical results}
\label{sect:numerics}

The $\qbar$-expansions can be used to obtain numerical results within the region of convergence of the series expansion.
We may verify the results by comparing to programs based on sector decomposition like 
\verb|sector_decomposition| \cite{Bogner:2007cr},
\verb|FIESTA| \cite{Smirnov:2008py,Smirnov:2009pb}
or
\verb|SecDec| \cite{Carter:2010hi,Borowka:2017idc,Borowka:2018goh}.

We start with $I_2^{(3)}$. The analytic expressions for this term involves only the (holomorphic)
modular forms $1$ and $f_{4,a}$.
It does not involve any meromorphic modular form. The $\qbar$-expansion converges therefore in the full complex upper half-plane. Translated to $x$-space this means that our result converges for all values 
\bq
 x & \in & {\mathbb R}\backslash\{0\} + i \delta.
\eq
\begin{figure}
\begin{center}
\includegraphics[scale=0.9]{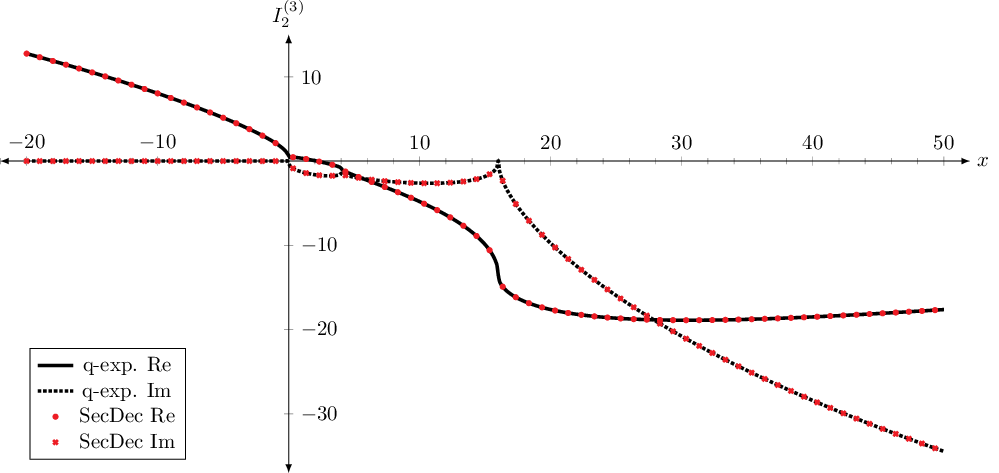}
\end{center}
\caption{
Comparison of our result for $I_2^{(3)}$ with numerical results from SecDec.
The $\qbar$-expansion converges for all points except $x=0$.
}
\label{fig_result_eps_3}
\end{figure}
The numerical values are shown in fig.~\ref{fig_result_eps_3}.
We also plotted the results from $\verb|SecDec|$.
We observe good agreement.
It is worth noting that $I_2^{(3)}$ as a function of $\tau$ is holomorphic in a neighbourhood of
$\tau=1/2+i\sqrt{3}/6$ and $\tau=1/4+i\sqrt{3}/12$.
The behaviour of $I_2^{(3)}$ at the threshold $x=16$ and the pseudo-threshold $x=4$ comes entirely from the kinks of the path in $\tau$-space
(see fig.~\ref{fig_contour_tau_qbar}).

Let us then look at the next order in the $\eps$-expansion.
The analytic result for $I_2^{(4)}$ involves the meromorphic modular forms, hence the $\qbar$-series converges
for
\bq
 x \in \left] -\infty, -2 \right[ + i \delta
 & \mbox{and} &
 x \in \left]16,\infty\right[ + i \delta.
\eq
\begin{figure}
\begin{center}
\includegraphics[scale=0.9]{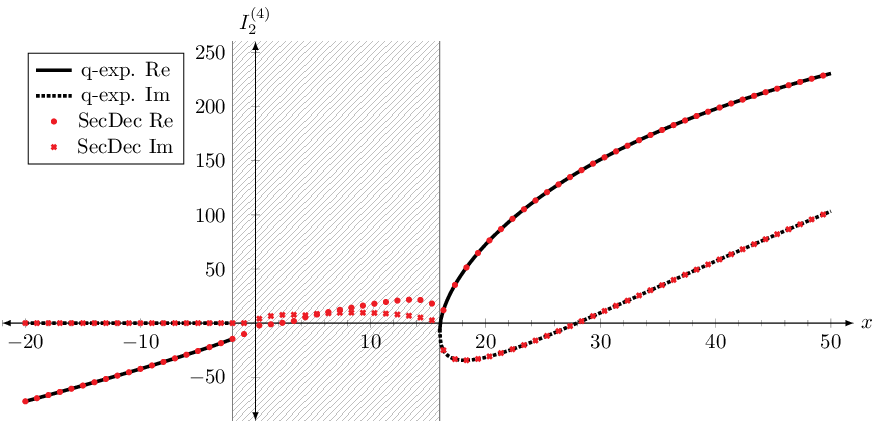}
\end{center}
\caption{
Comparison of our result for $I_2^{(4)}$ with numerical results from SecDec.
The $\qbar$-expansion converges for $x<-2$ and $x>16$. In the interval $[-2,16]$ only the points from SecDec are shown.
}
\label{fig_result_eps_4}
\end{figure}
The numerical results in these regions are shown in fig.~\ref{fig_result_eps_4}, again with the corresponding values from $\verb|SecDec|$.
We observe again good agreement.
In the interval $[-2,16]$ only the points from SecDec are shown.

% ------------------------------------------------------------------------------------------
\section{Conclusions}
\label{sect:conclusions}

In this paper we studied the three-loop banana integral with equal masses.
We derived a differential equation in $\eps$-factorised form, containing six letters.
The letters are built from meromorphic modular forms and one special function,
which can be given as an iterated integral of meromorphic modular forms. 
We investigated this special function in detail:
It has a $\qbar$-expansion with integer coefficients and only simple poles at $x=4$ and $x=16$.
Under modular transformations of $\Gamma_1(6)$ it transforms as ``quasi-Eichler'' (see eq.~(\ref{def_quasi_Eichler})).

The result of this paper adds further evidence that an $\eps$-factorised form of the differential equation might exist for any Feynman integral.

\subsection*{Acknowledgements}

We would like to thank Claude Duhr for useful discussions.

% ------------------------------------------------------------------------------------------
\begin{appendix}

% ------------------------------------------------------------------------------------------
\section{Boundary values}
\label{sect:boundary}

We determine the boundary at $1/x=0$.
We start from the Feynman parametrisation
\bq
 I_{1111}
 & =& 
 e^{3 \gamma_E \eps}
 \Gamma\left(1+3\eps\right)
 \int\limits_{a_i\ge 0} d^4a \; \delta\left(1-\sum\limits_{i=1}^4 a_i\right)
 \frac{{\mathcal U}^{4\eps}}{{\mathcal F}^{1+3\eps}},
\eq
with
\bq
 {\mathcal U} 
 & = & 
 a_1 a_2 a_3 + a_1 a_2 a_4 + a_1 a_3 a_4 + a_2 a_3 a_4,
 \nonumber \\
 {\mathcal F}
 & = &
 - x a_1 a_2 a_3 a_4 + \left(a_1+a_2+a_3+a_4\right) {\mathcal U}.
\eq
We follow the lines of \cite{Broedel:2019kmn} and exchange the Feynman parameter integration with a three-fold
Mellin-Barnes integration.
We arrive at
\bq
\lefteqn{
 I_{1111}
 = } & & \nonumber \\
 & &
 \frac{1}{2} \Gamma\left(\frac{1}{2}\right)
 2^{-4\eps}
 e^{3 \gamma_E \eps}
 \frac{1}{\left(2\pi i\right)^3}
 \int d\sigma_1
 \int d\sigma_2
 \int d\sigma_3
 \left(-\frac{4}{x}\right)^{-\sigma_1}
 4^{-\sigma_{23}}
 \Gamma\left(-\sigma_1\right)
 \Gamma\left(-\sigma_2\right)
 \Gamma\left(-\sigma_3\right)
 \nonumber \\
 & &
 \Gamma\left(-\sigma_2-\eps\right)
 \Gamma\left(-\sigma_3-\eps\right)
 \frac{\Gamma\left(\sigma_{123}+1+\eps\right)\Gamma\left(\sigma_{123}+1+2\eps\right)\Gamma\left(\sigma_{123}+1+3\eps\right)}
      {\Gamma\left(\sigma_{1}+1-\eps\right)\Gamma\left(\sigma_{123}+\frac{3}{2}+2\eps\right)},
\eq
with $\sigma_{ij}=\sigma_i+\sigma_j$ and $\sigma_{ijk}=\sigma_i+\sigma_j+\sigma_k$.
We are interested in the region where $x$ is large.
The half-circles at infinity vanish if we close the contour for the $\sigma_1$-integration to the left
and the contours for the $\sigma_2$- and $\sigma_3$-integration to the right.
Thus we pick up to residues of 
\bq
 \Gamma\left(\sigma_{123}+1+\eps\right), \;\;\; 
 \Gamma\left(\sigma_{123}+1+2\eps\right), \;\;\; 
 \Gamma\left(\sigma_{123}+1+3\eps\right)
\eq
for the $\sigma_1$-integration, the residues of 
\bq
 \Gamma\left(-\sigma_2\right), \;\;\;
 \Gamma\left(-\sigma_2-\eps\right)
\eq
for the $\sigma_2$-integration and the residues of
\bq
 \Gamma\left(-\sigma_3\right), \;\;\;
 \Gamma\left(-\sigma_3-\eps\right)
\eq
for the $\sigma_3$-integration.
This would give $12$ terms. However the triple residue of
\bq
 \frac{\Gamma\left(\sigma_{123}+1+\eps\right)\Gamma\left(-\sigma_2-\eps\right)\Gamma\left(-\sigma_3-\eps\right)}{\Gamma\left(\sigma_{1}+1-\eps\right)}
\eq
vanishes due to $\Gamma\left(\sigma_{1}+1-\eps\right)$ in the denominator, resulting in $11$ terms.
We obtain
\bq
\lefteqn{
 I_{1111}
 = 
 - \frac{2}{x}
 \Gamma\left(\frac{1}{2}\right)
 2^{-4\eps}
 e^{3 \gamma_E \eps}
 \sum\limits_{n_1=0}^\infty
 \sum\limits_{n_2=0}^\infty
 \sum\limits_{n_3=0}^\infty
 \frac{1}{n_1!n_2!n_3!}
 \left(\frac{4}{x}\right)^{n_1}
 \left(\frac{1}{x}\right)^{n_2+n_3}
} & & \nonumber \\
 & &
 \left\{
  \left(-\frac{4}{x}\right)^{\eps} \frac{\Gamma\left(-n_1+\eps\right)\Gamma\left(-n_1+2\eps\right)\Gamma\left(-n_2-\eps\right)\Gamma\left(-n_3-\eps\right)\Gamma\left(n_{123}+1+\eps\right)}{\Gamma\left(-n_{123}-2\eps\right)\Gamma\left(-n_1+\frac{1}{2}+\eps\right)}
 \right. \nonumber \\
 & & \left.
  + \left(-\frac{4}{x}\right)^{2\eps} \frac{\Gamma\left(-n_1-\eps\right)\Gamma\left(-n_1+\eps\right)\Gamma\left(-n_2-\eps\right)\Gamma\left(-n_3-\eps\right)\Gamma\left(n_{123}+1+2\eps\right)}{\Gamma\left(-n_{123}-3\eps\right)\Gamma\left(-n_1+\frac{1}{2}\right)}
 \right. \nonumber \\
 & & \left.
  + \left(-\frac{4}{x}\right)^{3\eps} \frac{\Gamma\left(-n_1-2\eps\right)\Gamma\left(-n_1-\eps\right)\Gamma\left(-n_2-\eps\right)\Gamma\left(-n_3-\eps\right)\Gamma\left(n_{123}+1+3\eps\right)}{\Gamma\left(-n_{123}-4\eps\right)\Gamma\left(-n_1+\frac{1}{2}-\eps\right)}
 \right. \nonumber \\
 & & \left.
  + 4^{\eps} \frac{\Gamma\left(-n_1+\eps\right)\Gamma\left(-n_1+2\eps\right)\Gamma\left(-n_2-\eps\right)\Gamma\left(-n_3+\eps\right)\Gamma\left(n_{123}+1\right)}{\Gamma\left(-n_{123}-\eps\right)\Gamma\left(-n_1+\frac{1}{2}+\eps\right)}
 \right. \nonumber \\
 & & \left.
  + \left(-\frac{16}{x}\right)^{\eps} \frac{\Gamma\left(-n_1-\eps\right)\Gamma\left(-n_1+\eps\right)\Gamma\left(-n_2-\eps\right)\Gamma\left(-n_3+\eps\right)\Gamma\left(n_{123}+1+\eps\right)}{\Gamma\left(-n_{123}-2\eps\right)\Gamma\left(-n_1+\frac{1}{2}\right)}
 \right. \nonumber \\
 & & \left.
  + \left(-\frac{8}{x}\right)^{2\eps} \frac{\Gamma\left(-n_1-2\eps\right)\Gamma\left(-n_1-\eps\right)\Gamma\left(-n_2-\eps\right)\Gamma\left(-n_3+\eps\right)\Gamma\left(n_{123}+1+2\eps\right)}{\Gamma\left(-n_{123}-3\eps\right)\Gamma\left(-n_1+\frac{1}{2}-\eps\right)}
 \right. \nonumber \\
 & & \left.
  + 4^{\eps} \frac{\Gamma\left(-n_1+\eps\right)\Gamma\left(-n_1+2\eps\right)\Gamma\left(-n_2+\eps\right)\Gamma\left(-n_3-\eps\right)\Gamma\left(n_{123}+1\right)}{\Gamma\left(-n_{123}-\eps\right)\Gamma\left(-n_1+\frac{1}{2}+\eps\right)}
 \right. \nonumber \\
 & & \left.
  + \left(-\frac{16}{x}\right)^{\eps} \frac{\Gamma\left(-n_1-\eps\right)\Gamma\left(-n_1+\eps\right)\Gamma\left(-n_2+\eps\right)\Gamma\left(-n_3-\eps\right)\Gamma\left(n_{123}+1+\eps\right)}{\Gamma\left(-n_{123}-2\eps\right)\Gamma\left(-n_1+\frac{1}{2}\right)}
 \right. \nonumber \\
 & & \left.
  + \left(-\frac{8}{x}\right)^{2\eps} \frac{\Gamma\left(-n_1-2\eps\right)\Gamma\left(-n_1-\eps\right)\Gamma\left(-n_2+\eps\right)\Gamma\left(-n_3-\eps\right)\Gamma\left(n_{123}+1+2\eps\right)}{\Gamma\left(-n_{123}-3\eps\right)\Gamma\left(-n_1+\frac{1}{2}-\eps\right)}
 \right. \nonumber \\
 & & \left.
  + 16^{\eps} \frac{\Gamma\left(-n_1-\eps\right)\Gamma\left(-n_1+\eps\right)\Gamma\left(-n_2+\eps\right)\Gamma\left(-n_3+\eps\right)\Gamma\left(n_{123}+1\right)}{\Gamma\left(-n_{123}-\eps\right)\Gamma\left(-n_1+\frac{1}{2}\right)}
 \right. \nonumber \\
 & & \left.
  + \left(-\frac{64}{x}\right)^{\eps} \frac{\Gamma\left(-n_1-2\eps\right)\Gamma\left(-n_1-\eps\right)\Gamma\left(-n_2+\eps\right)\Gamma\left(-n_3+\eps\right)\Gamma\left(n_{123}+1+\eps\right)}{\Gamma\left(-n_{123}-2\eps\right)\Gamma\left(-n_1+\frac{1}{2}-\eps\right)}
 \right\}.
\eq
For the boundary value we are only interested in the leading term (with all logarithms) in an $1/x$-expansion.
This leading term is given by setting $n_1=n_2=n_3=0$ in the expression above.
We obtain
\bq
\lefteqn{
 I_{1111}
 = 
 - \frac{2}{x}
 e^{3 \gamma_E \eps}
 \left\{
  2 \Gamma\left(\eps\right)^3
  + 3 \left(-\frac{1}{x}\right)^{\eps} \frac{\Gamma\left(-\eps\right)^2\Gamma\left(\eps\right)^2\Gamma\left(1+\eps\right)}{\Gamma\left(-2\eps\right)}
 \right.
} & & \\
 & & \left.
  + 2 \left(-\frac{1}{x}\right)^{2\eps} \frac{\Gamma\left(-\eps\right)^3\Gamma\left(\eps\right)\Gamma\left(1+2\eps\right)}{\Gamma\left(-3\eps\right)}
  + \frac{1}{2} \left(-\frac{1}{x}\right)^{3\eps} \frac{\Gamma\left(-\eps\right)^4\Gamma\left(1+3\eps\right)}{\Gamma\left(-4\eps\right)}
 \right\}
 + {\mathcal O}\left(x^{-2}\right).
 \nonumber
\eq
This agrees with the first part of eq. (A.8) in \cite{Broedel:2019kmn}.
In the limit $1/x\rightarrow 0$ we further have
\bq
 \frac{\omega_1}{\pi^2} \; = \; - \frac{12}{x}
 + {\mathcal O}\left(x^{-2}\right)
 & \mbox{and} &
 \qbar \; = \; - \frac{1}{x}
 + {\mathcal O}\left(x^{-2}\right).
\eq
We find
\bq
 I_2
 & = &
 \left[ \frac{4}{3} \zeta_3 - \frac{1}{3}\ln^3\left(\qbar\right) \right] \eps^3 
 + \left[ 2 \zeta_4 + 10 \zeta_3 \ln\left(\qbar\right) + \zeta_2 \ln^2\left(\qbar\right) - \frac{1}{2} \ln^4\left(\qbar\right) \right] \eps^4
 \nonumber \\
 & &
 + \left[ 
          36 \zeta_5 - 8 \zeta_2 \zeta_3
          + \frac{25}{2} \zeta_4 \ln\left(\qbar\right)
          + 23 \zeta_3 \ln^2\left(\qbar\right)
          + \frac{3}{2} \zeta_2 \ln^3\left(\qbar\right)
          - \frac{5}{12} \ln^5\left(\qbar\right)
   \right] \eps^5
 \nonumber \\
 & &
 + \left[ 
          \frac{537}{8} \zeta_6 - 84 \zeta_3^2
          + \left( 150 \zeta_5 - 31 \zeta_2 \zeta_3 \right) \ln\left(\qbar\right)
          + \frac{125}{4} \zeta_4 \ln^2\left(\qbar\right)
          + \frac{88}{3} \zeta_3 \ln^3\left(\qbar\right)
 \right. \nonumber \\
 & & \left.
          + \frac{4}{3} \zeta_2 \ln^4\left(\qbar\right)
          - \frac{1}{4} \ln^6\left(\qbar\right)
   \right] \eps^6
 + {\mathcal O}\left(\qbar\right)
 + {\mathcal O}\left(\eps^7\right).
\eq
This corrects the second part of eq. (A.8) in \cite{Broedel:2019kmn}.
(The expansion is of uniform weight, weight drops do not occur.)

% -----------------------------------------------------------------------------

\section{Supplementary material}
\label{sect:supplement}

Attached to the arxiv version of this article is an electronic file in ASCII format with {\tt Maple} syntax, defining the quantities
\begin{center}
 \verb|I_symb|, \; \verb|I_qbar|.
\end{center}
The vector \verb|I_symb| contains the results for the master integrals up to order $\eps^6$ in terms of iterated integrals.
The vector \verb|I_qbar| contains the results for the master integrals up to order $\eps^6$ as an expansion in $\qbar$ up to order $\qbar^{30}$.

\end{appendix}

% ----------------------------------------------
% references
\bibliography{/home/stefanw/notes/biblio}
\bibliographystyle{/home/stefanw/latex-style/h-physrev5}

\end{document}